\documentclass[lettersize,journal]{IEEEtran}
\usepackage{amsmath,amsfonts,amsthm}
\allowdisplaybreaks[4]
\usepackage{amssymb}
\usepackage{algpseudocode}
\usepackage{algorithm}
\usepackage{array}
\usepackage{textcomp}
\usepackage{stfloats}
\usepackage{url}
\usepackage{verbatim}
\usepackage{graphicx}
\usepackage{subcaption}
\usepackage{cite}
\usepackage{hyperref}
\usepackage{xcolor}
\usepackage{booktabs}
\usepackage{multirow}
\newtheoremstyle{mystyle}
{3pt}
{3pt}
{}
{1em}
{\bfseries\itshape}
{:}
{0.5em}
{}

\hypersetup{
    colorlinks=true,
    citecolor=blue,
    linkcolor=magenta,
}

\begin{document}

\theoremstyle{mystyle}   % 设置定理样式为mystyle

% 1. proof Env prefix (default "Proof")
\providecommand{\proofname}{Proof}
% 2. redefine proof env
\makeatletter
\renewenvironment{proof}[1][\proofname]{\par
\pushQED{\qed}%
\normalfont \topsep3\p@\@plus3\p@\relax
\trivlist
\item\relax
% default \itshape
\quad\quad
{\itshape
#1\@addpunct{:}}\hspace\labelsep\ignorespaces
}{%
\popQED\endtrivlist\@endpefalse
}
\makeatother

\newtheorem{assumption}{Assumption}
\newtheorem{theorem}{Theorem}
\newtheorem{lemma}{Lemma}
\newtheorem{proposition}{Proposition}
\newtheorem{remark}{Remark}

%\setlength{\abovedisplayskip}{1.2pt}
%\setlength{\abovedisplayshortskip}{1.2pt}
%\setlength{\belowdisplayskip}{1.2pt}
%\setlength{\belowdisplayshortskip}{1.2pt}
%
%% Compact float and caption spacing
%\setlength{\textfloatsep}{5pt plus 1pt minus 2pt}
%\setlength{\floatsep}{4pt plus 1pt minus 1pt}
%\setlength{\intextsep}{5pt plus 1pt minus 2pt}
%\setlength{\dbltextfloatsep}{5pt plus 1pt minus 2pt}
%\setlength{\dblfloatsep}{4pt plus 1pt minus 1pt}
%
%
%% Compact multi-line equations
%\setlength{\jot}{1pt}
%\setlength{\arraycolsep}{2pt}
%\setlength{\tabcolsep}{3pt}

\title{Sensing-Native Over-the-Air Federated Learning}

\author{Peiyuan Huang,~Shijian Gao,~\IEEEmembership{Member, IEEE},~Jia Yan,~\IEEEmembership{Member, IEEE},~Georgios B. Giannakis,~\IEEEmembership{Life Fellow, IEEE}
        % <-this % stops a space
\thanks{This work was presented in part at the IEEE International Geoscience and Remote Sensing Symposium (IGARSS), Pasadena, CA, USA, 2023~\cite{gao2023integrated}.}%
\thanks{Peiyuan Huang and Jia Yan are with the Intelligent Transportation Thrust, The Hong Kong University of Science and Technology (Guangzhou), Guangzhou 511453, China (e-mails: phuang169@connect.hkust-gz.edu.cn; jasonjiayan@hkust-gz.edu.cn).}%
\thanks{Shijian Gao is with the Internet of Things Thrust, The Hong Kong University of Science and Technology (Guangzhou), Guangzhou 511400, China (e-mail: shijiangao@hkust-gz.edu.cn).}%
\thanks{Georgios B. Giannakis is with the Department of Electrical and Computer Engineering, University of Minnesota, Minneapolis MN 55455 USA (e-mail: georgios@umn.edu).}
}

%\markboth{Journal of \LaTeX\ Class Files,~Vol.~14, No.~8, August~2021}


\maketitle

\begin{abstract}
Over-the-air federated learning (FL) leverages the superposition property of multiple-access channels to enable communication-efficient distributed model training. Existing integrated sensing, communication, and computation (ISCC)-enabled over-the-air FL systems typically require dedicated resources for the sensing module, inevitably compromising FL performance due to resource competition. \textcolor{black}{In this paper, we propose a sensing-native over-the-air FL framework that explores built-in distributed wireless sensing capability with zero overhead per model aggregation.
%enables target localization without additional sensing overhead by exploiting the distributed wireless sensing functionality inherently embedded in each model aggregation round. 
%Specifically, we reveal that the signal of each wireless device (WD) used to convey FL local updates possesses inherent auto-correlation characteristics that make it well-suited for sensing.
Specifically, the high-dimensional local gradient signals possessing favorable autocorrelation property are concurrently leveraged for target distance estimation, while the gradient statistics already required for over-the-air FL serve as a ready-made gateway to deliver locally-sensed results to the edge server for cooperative localization. To combat inter-device interference, channel fading, and communication noise, we put forth a robust trilateration-based target positioning method building upon an efficient matched-filtering-based distance estimation.
%each wireless device estimates its target distance via matched filtering of the echo gradient signal, and the BS performs robust trilateration with density-based clustering. 
Then, by explicitly characterizing the impact of imperfect model aggregation and noisy gradient-statistics transmission on the sensing-native over-the-air FL convergence, we develop a statistics-aware communication-learning co-design approach.} We first derive the closed-form optimal power budgets
allocated to local gradients and their statistics, based on which an
efficient successive convex approximation method is proposed for
receiver beamforming optimization. Simulation results show that the proposed framework simultaneously achieves superior learning and sensing performance  compared to representative baselines.
\end{abstract}


\section{Introduction}
Federated learning (FL) is a promising distributed edge learning paradigm that allows multiple wireless devices (WDs) to collaboratively learn a shared model under the coordination of a base station (BS) and without exchanging data directly \cite{mcmahan2017communication}. The need to transmit high-dimensional local model parameters for global aggregation substantially undermines the communication efficiency of FL over wireless networks. To tackle this issue, over-the-air computation \cite{nazer2007computation} allows simultaneous transmissions of all device-end updates through a shared spectrum by leveraging the signal-superposition property of multiple-access channels. 

With integrated communication and computation, over-the-air FL \cite{zhu2019broadband} has great potential for substantial resource savings while maintaining comparable performance to conventional FL with orthogonal transmission. The communication-learning co-design for over-the-air FL has been extensively investigated in \cite{yang2020federated,sery2021over,xu2021learning,9606731,9681882,r1} to mitigate the impact of wireless channel fading and communication noise on the FL performance by simply assuming error-free transmission of learning statistics, i.e., means and variances of local gradients in each training round. For example, \cite{r1} exploits reconfigurable intelligent surface to enhance the channels of stragglers and thereby improve FL performance, while its transceiver design relies on the precise gradient statistics reported by the WDs.

Beyond the synergy of communications and computing, the geographically diverse deployment of WDs and the high-dimensional local update vectors with desirable autocorrelation properties further endow over-the-air FL with distributed wireless sensing capability, an aspect that has been largely overlooked in the existing literature. \textcolor{black}{In return, such inherent sensing capability equips over-the-air FL with environmental awareness to perceive blockage, mobility, and surrounding dynamics, thereby facilitating environment-adaptive client scheduling and resource allocation for improved learning performance.} In this paper, we propose a novel sensing-native over-the-air FL framework by harnessing the inherent distributed sensing functionality along with collaborative model training. \textcolor{black}{The native sensing in our framework exploits not only the local gradient signals with favorable autocorrelation property but also the information delivery infrastructure of over-the-air FL, making it a zero-overhead built-in capability of the model aggregation process.
%where sensing is treated as a built-in capability of the over-the-air model aggregation process rather than an additional function that requires dedicated resource consumption.
} 
Specifically, during per-round over-the-air model aggregation, each WD first analyzes the echo signals reflected from the passive sensing target to obtain the corresponding distance. Accordingly, the learning statistics act as an off-the-shelf gateway to deliver locally-sensed results to the BS for cooperative localization without introducing additional overhead. The proposed framework enables the coherent fusion of sensing, communications, and computation in a wireless manner.

It is worth mentioning that the proposed sensing-native over-the-air FL framework is fundamentally different from the existing works on ISCC-enabled FL systems \cite{11241079,liu2025uav,tang2025integrated,liu2022toward,liu2024multi,hu2025differentially,10767214}. Prior art requires dedicated resources to be allocated to the wireless sensing module. These resources, however, must compete with those demanded by the communications and computing modules. For instance, in \cite{liu2022toward}, WDs in the human motion recognition scenario leverage extra radio resources for wireless sensing to continuously collect training samples and enrich the local dataset throughout the FL process. The joint optimization of sensing, computation, and communication resources is studied to maximize the FL convergence speed, where the total number of sensed samples and the round-varying batch sizes are efficiently obtained. \cite{liu2024multi} investigated a multi-objective optimization over sensing, communication, and computation resources, by involving Cramer-Rao Bound, sum rate and power consumption as the performance metrics. Considering over-the-air FL, the BS performs sensing and FL in the same time-frequency resource \cite{du2024integrated}, where the transceiver beamforming and device selection are jointly optimized by minimizing the FL convergence gap while satisfying the sensing requirement. \cite{10615482} classifies WDs into three groups:  those selected to participate in over-the-air FL while reusing their uplink signals for sensing, those activated solely to transmit dedicated sensing waveforms, and those remaining idle. Accordingly, a joint client selection and power control problem is studied to strike a balance between learning and sensing performance. Unlike the existing ISCC-empowered over-the-air FL works, where dedicating resources to wireless sensing inevitably compromises FL performance due to resource competition, this paper explores the inherent wireless sensing functionality embedded in each over-the-air FL model aggregation round at no additional cost, unlocking performance gains in both tasks.

The design challenges of such a sensing-native over-the-air FL framework are twofold. First, the echo signal reflected by the sensing target not only contains the highly correlated local gradient signal but is also contaminated by inter-device interference due to over-the-air aggregation, rendering the local distance estimation at each WD inaccurate. Moreover, the uplink transmission of the locally estimated distance information, embedded in the local learning statistics, is further perturbed by wireless fading and communication noise, posing a significant challenge to collaborative target localization at the BS. Therefore, the first question raised is how to design an accurate distributed wireless sensing method by leveraging only the local gradient signals over the air at zero overhead. 

Second, the learning statistics signals naturally serve as a ready-made enabler that simultaneously conveys the locally sensed distance information and the statistical moments  of local gradients. Although requiring negligible additional frequency resources, the imperfect transmission of such dual-functional signals leads to severe sub-optimality in FL learning performance. This gives rise to the statistics-aware communication-learning co-design problem: how to maximize the overall sensing-native over-the-air FL performance in the presence of noisy learning statistics per-round. 

In this paper, we propose a novel sensing-native over-the-air FL framework to free up the entirety of ISCC resources for the improvement of learning performance while simultaneously enabling distributed zero-overhead wireless sensing with high accuracy. The main contributions are summarized as follows:

\begin{itemize}
    \item The proposed framework integrates sensing, communication, and computation in a fully native manner by reusing the uplink gradient signals and the associated learning-statistics transmission already required by over-the-air FL. This design enables cooperative target localization without dedicated sensing waveforms or additional radio resources, thereby eliminating the performance trade-off between sensing and FL.
    \item We develop a zero-overhead cooperative localization method tailored to the proposed framework. Each WD estimates its target distance from the echo of its local gradient signal via matched filtering, while the BS fuses the uploaded distance estimates for target localization. To combat the inter-user interference inherent in over-the-air FL, we further design an efficient trilateration method with density-based clustering for localization robustness.
    
    \item %\textcolor{black}{This work is among the first to explicitly characterize the coupling effects of imperfect model aggregation and noisy transmission of gradient statistics on the convergence of the proposed sensing-native over-the-air FL framework. We reveal that the imperfect  statistics transmission amplifies over-the-air aggregation error, calling for the joint optimization of receiver beamforming and transmit power budgets allocated to local gradients and their statistics.  }
    %a previously overlooked coupling effect between these two impairments. This coupling effect shows that inaccurate variance estimation used for transmit-scaling design can further amplify over-the-air model aggregation error. This result uncovers a new source of convergence degradation and highlights the necessity of jointly accounting for local-gradient and statistics transmission in the learning analysis.
    This work is among the first to explicitly characterize the coupling effects of imperfect model aggregation and noisy gradient-statistics transmission on the convergence of the proposed sensing-native over-the-air FL. We reveal that imperfect statistics transmission amplifies over-the-air aggregation error, calling for the joint optimization of receiver beamforming and transmit power budgets allocated to local gradients and their statistics.
    \item %We develop a statistics-aware communication-learning co-design approach 
    %by jointly optimizing transmit power allocation and receive beamforming strategy to minimize the theoretically-derived convergence gap. The problem is challenging due to the non-convex nature and strong coupling among diverse decision variables. We first derive the optimal transmit power budgets allocated to the local gradient and its statistics in closed forms, wherein the engineering insights can be observed for practical implementation. Accordingly, the optimal beamforming scheme is obtained by an efficient successive convex approximation method.
    We develop a statistics-aware communication-learning co-design approach to enhance the learning convergence. Despite the non-convexity and strong coupling among decision variables, we first derive the optimal transmit power budgets allocated to local gradient and its statistics in closed forms, based on which the optimal beamforming scheme is obtained by an efficient successive convex approximation method.
\end{itemize}

Extensive simulation results demonstrate that the proposed sensing-native over-the-air FL framework achieves superior learning accuracy while simultaneously attaining broader sensing coverage with lower target localization error, compared to the representative baselines.

The remainder of this paper is organized as follows. Section \ref{sec_sysmodel} introduces the proposed novel sensing-native over-the-air FL framework. Section \ref{sec_sensing} details the zero-overhead cooperative localization method tailored to the proposed framework. In Section \ref{sec_analysis}, we present the learning performance analysis for the proposed framework and formulate the statistics-aware communication-learning co-design problem. Section \ref{sec_opt} develops efficient approaches to jointly optimize the transmit power budget allocation and receiver beamforming. Section \ref{sec_exp} evaluates the performance of the proposed framework through extensive simulations. Section \ref{sec_con} concludes this paper.

\section{System Model}\label{sec_sysmodel}

\subsection{System Overview}

As shown in Fig. \ref{fig_model}, we propose the sensing-native over-the-air FL framework comprising $M$ single-antenna WDs, indexed by $m \in \mathcal{M}=\{1,2, \dots, M\}$, and an $N$-antenna BS. The target is to collaboratively train a global model through FL, along with zero-overhead wireless sensing by harnessing the transmit signals of local updates during model aggregations. Specifically, the $D$-dimensional global model $\mathbf{w}$ is optimized by minimizing the empirical global loss function, i.e., 
\begin{align}\label{global_loss}
\min_{\mathbf{w}}F(\mathbf{w})=\frac1K\sum_{k=1}^Kf(\mathbf{w};(\mathbf{u}_{k},v_{k})), 
\end{align}
where $f(\mathbf{w};(\mathbf{u}_{k},v_{k}))$ is the loss function with respect to the $k$-th training sample $(\mathbf{u}_{k},v_{k})$. $\mathbf{u}_{k}$ and $v_{k}$ are the input feature vector and the output label, respectively.
Suppose that the total number of $K$ training samples are distributed among WDs, with the $m$-th WD holding local dataset $\mathcal{D}_m=\{(\mathbf{u}_{m,k},v_{m,k}), 1\leq k\leq K_m\}$. 

To conduct FL, WD $m$ first performs local training on $\mathcal{D}_m$ by minimizing the local empirical loss function $F_m(\mathbf{w})= \frac{1}{K_m}\sum_{k=1}^{K_m}f(\mathbf{w};(\mathbf{u}_{m,k},v_{m,k}))$.
Then, WDs upload their local training results to the BS for global model update. Accordingly, the global loss function in (\ref{global_loss}) can be rewritten as a weighted sum of the local loss functions, i.e., $F(\mathbf{w})=\frac{1}{\sum_{m =1}^{M} K_m}\sum_{m =1}^{M} K_mF_m(\mathbf{w})$.
We perform a distributed gradient descent algorithm to solve problem (\ref{global_loss}) with $T$-round model aggregations, where the local update signals of WDs per iteration are naturally exploited for target sensing. Specifically, the $t$-th training iteration contains the following steps:
\begin{enumerate}
    \item \textbf{Model broadcasting}: The BS broadcasts the current global model $\mathbf{w}_t$ to all WDs.
    
    \item \textbf{Local training}: Each WD computes its local gradient using the local dataset by adopting batch gradient descent, i.e., $\mathbf{g}_{m,t} \!=\! \nabla F_m(\mathbf{w}) \!=\! \frac{1}{K_m}\!\sum_{k=1}^{K_m}\nabla f(\mathbf{w};(\mathbf{u}_{m,k},v_{m,k})).$
%\begin{align}
%    \end{align}
\item \textbf{Model aggregation}: WDs transmit $\mathbf{g}_{m,t}$ to the BS via wireless channels. To update the global model, the BS intends to calculate the weighted sum of gradients $\mathbf{r}_t=\sum_{m =1}^{M} K_m\mathbf{g}_{m,t}$. However, due to the channel fading and communication noise, only the estimate of $\mathbf{r}_t$, denoted as $\hat{\mathbf{r}}_t$, can be obtained. Then, the BS updates the global model by $\mathbf{w}_{t+1}=\mathbf{w}_t-\frac{\lambda}{\sum_{m =1}^{M} K_m} \hat{\mathbf{r}}_t$,
where $\lambda$ denotes the learning rate.
    
\item \textbf{Target sensing}: The model aggregation process inherently enables distributed target sensing by exploiting the favorable autocorrelation property of local gradient signals. Specifically, each WD $m$ first analyzes the echo signals reflected from the passive target and intends to acquire the corresponding distance $R_{m}$. However, due to communication imperfections and the interference from the other WDs' local gradient signals, one can only obtain the estimate of $R_{m}$ per round $t$, defined as $\tilde{R}_{m,t}$, which is subsequently
uploaded to the BS for cooperative localization.
\end{enumerate}

\begin{figure}[t]\centering
\includegraphics[width=0.42\textwidth]{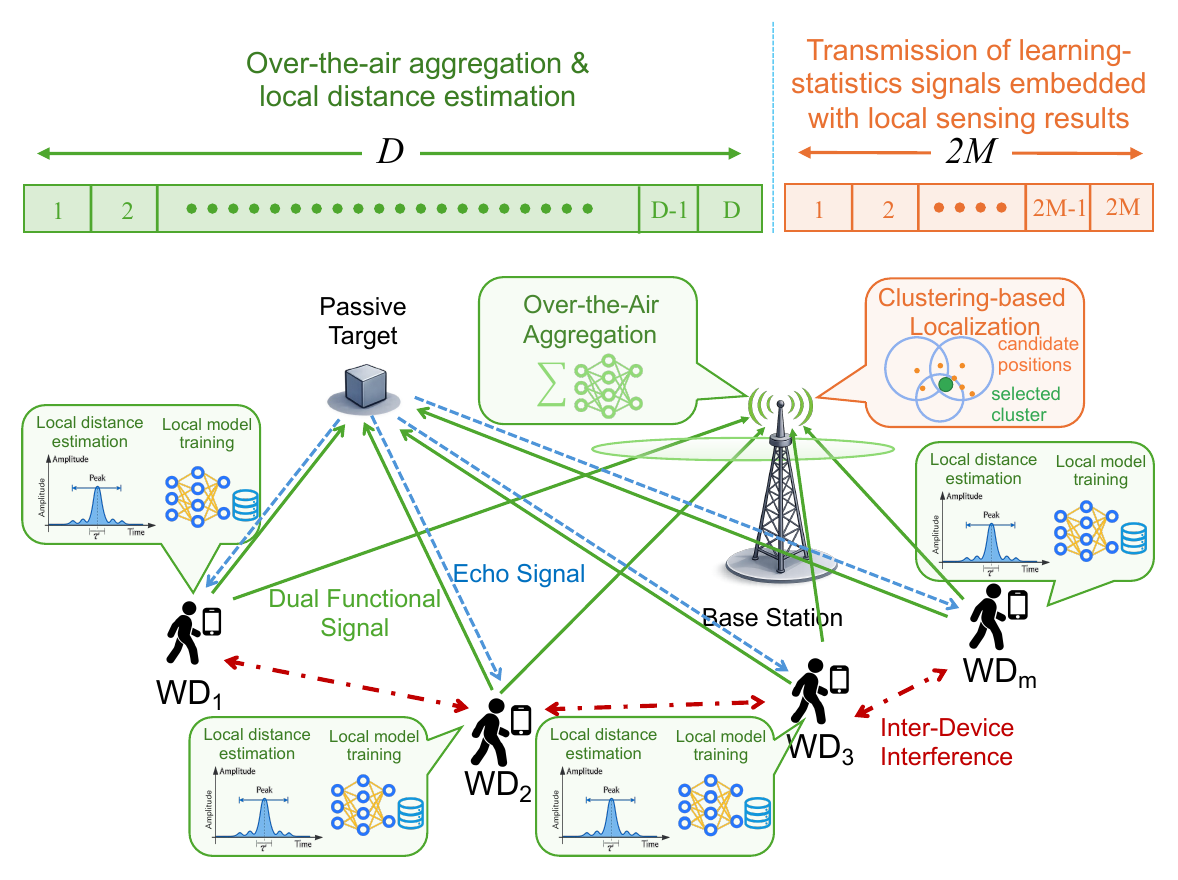}
% where an .eps filename suffix will be assumed under latex, 
% and a .pdf suffix will be assumed for pdflatex; or what has been declared
% via \DeclareGraphicsExtensions.
\caption{The proposed sensing-native over-the-air FL framework.}
\label{fig_model}
\end{figure}

\subsection{Dual-functional Transmit Signal Design}

To simultaneously perform model aggregation and target sensing, each WD $m$ first converts local gradient $\mathbf{g}_{m,t}$ to the transmit signal $\mathbf{x}_{1,m,t}$. Specifically, we first calculate the mean and standard deviation of local gradient vector $\mathbf{g}_{m,t}$ by 
\begin{align}
    \mu_{m,t}\!=\! \frac{1}{D} \!\sum_{d=1}^Dg_{m,t}[d], \ \nu_{m,t}^2 \!=\! \frac{1}{D} \!\sum_{d=1}^D\left(g_{m,t}[d]\!-\!\mu_{m,t}\right)^2,
\end{align}
where $g_{m,t}[d]$ is the $d$-th entry of $\mathbf{g}_{m,t}$. Then, the $d$-th entry of the transmit signal $\mathbf{x}_{1,m,t}$ is given by
\begin{align}\label{tx_grad}
    x_{1,m,t}[d] =p_{1,m,t} \bar{g}_{m,t}[d],
\end{align}
where $\bar{g}_{m,t}[d] = \frac{g_{m,t}[d] - \mu_{m,t}}{\nu_{m,t}}$ is the normalized local gradient, $p_{1,m,t} \in \mathbb{C}$ is the precoding factor that combats channel rotation and controls the transmit power. Here, we consider the power constraint for local gradient transmission as
\begin{align}
    \mathbb{E}[|x_{1,m,t}[d]|^2] &= |p_{1,m,t}|^2 \mathbb{E}\left[\left|\frac{g_{m,t}[d] - \mu_{m,t}}{\nu_{m,t}}\right|^2\right] \leq P_g. \label{pg_c}
\end{align}
To upload the local gradient statistics (i.e., mean $\mu_{m,t}$ and variance $\nu_{m,t}$) from each WD to the BS, we first design the transmit signal for the gradient variance $\nu_{m,t}$ of WD $m$ as
\begin{align}\label{tx_nu}
    x_{2,m,t} &= \frac{\sqrt{p_{2,m,t}}}{\nu_{\text{max}}}|\nu_{m,t}|,
\end{align}
where $\nu_{\text{max}}$ is the given maximum gradient variance, and $p_{2,m,t}$ is the transmit power for the gradient variance. Accordingly, we denote the corresponding peak transmit power as $P_\nu$ with 
\begin{align}
    \left|x_{2,m,t}\right|^2 = p_{2,m,t} \left(\frac{\nu_{m,t}}{\nu_{\text{max}}}\right)^2 \leq P_\nu \label{pv_c}.
\end{align}
Besides the local gradient and the associated gradient statistics, each WD also needs to offload locally-estimated distance of the sensing target to the BS for collaborative localization per FL round. In this paper, we embed the estimated target distance information into the imaginary part of the gradient mean signal, i.e.,
\begin{align}
    x_{3,m,t} =\frac{\sqrt{p_{3,m,t}}}{\sqrt{2}}\left(\frac{1}{\mu_{\text{max}}}\mu_{m,t}+j\frac{1}{R_{\text{max}}}\tilde{R}_{m,t}\right) \label{tx_ud},
\end{align}
where $\mu_{\text{max}}$ is the maximum gradient mean, $R_{\text{max}}$ signifies the maximum sensing range, and $p_{3,m,t}$ is the corresponding transmit power. Then, the associated peak power constraint is 
\begin{align}
    \left|x_{3,m,t}\right|^2 = \frac{p_{3,m,t}}{2} \left|\frac{\mu_{m,t}}{\mu_{\text{max}}}+j\frac{\tilde{R}_{m,t}}{R_{\text{max}}}\right|^2 \leq P_\mu \label{pu_c}.
\end{align}
The total transmit power of each WD per round is restricted by 
\begin{align}
    P_G+P_\mu+P_\nu\leq P_{\text{max}} \label{eq:budge},
\end{align}
where $P_G=DP_g$ is the power for transmitting $D$-dimension local gradient vector.

Notice that by harnessing the inherent autocorrelation property of local gradient signals $\mathbf{x}_{1,m,t}$ for target sensing and subsequently embedding the local sensing results into the gradient statistic signal, we achieve the zero-overhead wireless sensing during the FL process. The detailed sensing model and the proposed cooperative target localization method will be introduced in Section \ref{senmodel} and Section \ref{sec_sensing}, respectively. In addition, unlike existing works assuming error-free transmission of local gradient statistics, we will show in this paper that, the selection of $P_\nu$ and $P_\mu$ has an essential impact on both the learning and sensing performance, which necessitates the joint optimization with the transceiver for per-round model aggregation.

\subsection{Reception of Sensing Results and Over-the-Air Model Aggregation}

The wireless channel coefficient between WD $m$ and the BS is denoted as $\mathbf{h}_{bs,m} \in \mathbb{C}^{N \times 1}$. Suppose that the signals $x_{2,m,t}$ and $x_{3,m,t}$ carrying both local gradient statistics and local sensing results are transmitted to the BS through orthogonal channels. Then, the corresponding received signals $\mathbf{y}_{2,m,t}$ and $\mathbf{y}_{3,m,t}$ are given by 
\begin{align}
\mathbf{y}_{2,m,t} &= \mathbf{h}_{bs,m} x_{2,m,t} + \boldsymbol{\omega}_{2,m,t}, \\\mathbf{y}_{3,m,t} &= \mathbf{h}_{bs,m} x_{3,m,t} + \boldsymbol{\omega}_{3,m,t}, 
\end{align}
where $\boldsymbol{\omega}_{2,m,t}$ and $\boldsymbol{\omega}_{3,m,t}$ represent the additive white Gaussian noise (AWGN) following $\mathcal{CN}(\boldsymbol{0},\sigma^2\boldsymbol{I}) \in \mathbb{C}^{N \times 1}$.

With coherent detection, the BS post-processes the received signals by
\begin{align}
    r_{2,m,t}&\!=\!\frac{\nu_{\text{max}}\mathbf{f}^H\mathbf{y}_{2,m,t}}{\sqrt{p_{2,m,t}}\mathbf{f}^H\mathbf{h}_{bs,m}}\!=\!|\nu_{m,t}|\!+\!\frac{\nu_{\text{max}}\mathbf{f}^H\boldsymbol{\omega}_{2,m,t}}{\sqrt{p_{2,m,t}}\mathbf{f}^H\mathbf{h}_{bs,m}}, \label{r2_trans}\\
    r_{3,m,t}&\!=\!\frac{\mu_{\text{max}}\mathbf{f}^H\mathbf{y}_{3,m,t}}{\sqrt{p_{3,m,t}}\mathbf{f}^H\mathbf{h}_{bs,m}} \notag\\
    &\!=\!\frac{1}{\sqrt{2}}\!\left(\!\mu_{m,t}+j\frac{{\mu}_{\text{max}}}{R_{\text{max}}}\tilde{R}_{m,t}\!\right)\!+\!\frac{\mu_{\text{max}}\mathbf{f}^H\boldsymbol{\omega}_{3,m,t}}{\sqrt{p_{3,m,t}}\mathbf{f}^H\mathbf{h}_{bs,m}},\label{r3_trans}
\end{align}
where $\mathbf{f} \in \mathbb{C}^{N \times 1}$ is the normalized received beamforming vector with $||\mathbf{f}||^2_2=1$. Accordingly, the BS estimates the gradient statistics as $\hat{\nu}_{m,t}=\left|r_{2,m,t}\right|\label{nu_hat}$ and $ \hat{\mu}_{m,t}=\sqrt{2}\text{Re}\{r_{3,m,t}\}$.
Along with the local gradient statistics estimation, the BS can concurrently detect the local sensing results as 
\begin{align}
\hat{R}_{m,t}&=\sqrt{2}\frac{R_{\text{max}}}{\mu_{\text{max}}}\text{Im}\{r_{3,m,t}\} .\label{r_hat}
\end{align}
Given that the BS is interested in the weighted sum of the local gradients instead of the individual ones, WDs exploit the superposition property of wireless multi-access channels to simultaneously transmit amplitude-modulated local gradient signals $\{x_{1,m,t}[d]\}$ entry by entry using the same time-frequency communication resources. Accordingly, the transmission of the local gradients, the corresponding gradient statistics, and the locally sensed distance information requires a total of $D+2M$ time slots per training round.\footnote{Notice that existing works on over-the-air FL typically overlook the $2M$ time slots allocated for the transmission of local gradient statistics. This paper exploits the inherent wireless sensing functionality embedded within the original resource budget of over-the-air FL, i.e., 
$D+2M$ slots per round, at no additional cost. Specifically, we leverage the $D$-dimensional local gradient signals to perceive the target distance at each WD, and the locally estimated distances are then embedded into the learning-statistics signals for cooperative localization at the BS.   %Note that ``zero-overhead sensing’’ means that no extra time slots are introduced for sensing, since the BS already requires $2M$ slots to collect the WDs’ gradient statistics, and the locally estimated distances are embedded into the statistic signals for uplink transmission.
} 
Suppose that the length of each time slot is $T_{slot}$. With over-the-air model aggregation, the received signal $\mathbf{y}_{1,t}[d]$ at round $t$ for each time slot $d, 1\leq d\leq D$, is
\begin{align}\label{received_signal}
    \mathbf{y}_{1,t}[d] = \sum_{m=1}^{M}{\mathbf{h}_{bs,m} x_{1,m,t}[d]} + \boldsymbol{\omega}_{1,t} ,\ \ 1\leq d\leq D,
\end{align}
where $\boldsymbol{\omega}_{1,t}$ is the AWGN following $\mathcal{CN}(\boldsymbol{0},\sigma^2\boldsymbol{I}) \in \mathbb{C}^{N \times 1}$. Then, the BS estimates the weighted sum of local gradients $\hat{r}_{t}[d]$ by a linear estimator, i.e., 
\begin{align} \label{key1}
    \hat{r}_{t}[d] &= \frac{1}{\sqrt{\eta_t}} \mathbf{f}^H \mathbf{y}_{1,t}[d] +\sum_{m=1}^{M} K_m \hat{\mu}_{m,t},
\end{align}
where $\eta_t$ is the normalization scalar at round $t$.

\subsection{Sensing Model}\label{senmodel}
\begin{figure}[!t]
    \centering
    \includegraphics[width=0.45\textwidth]{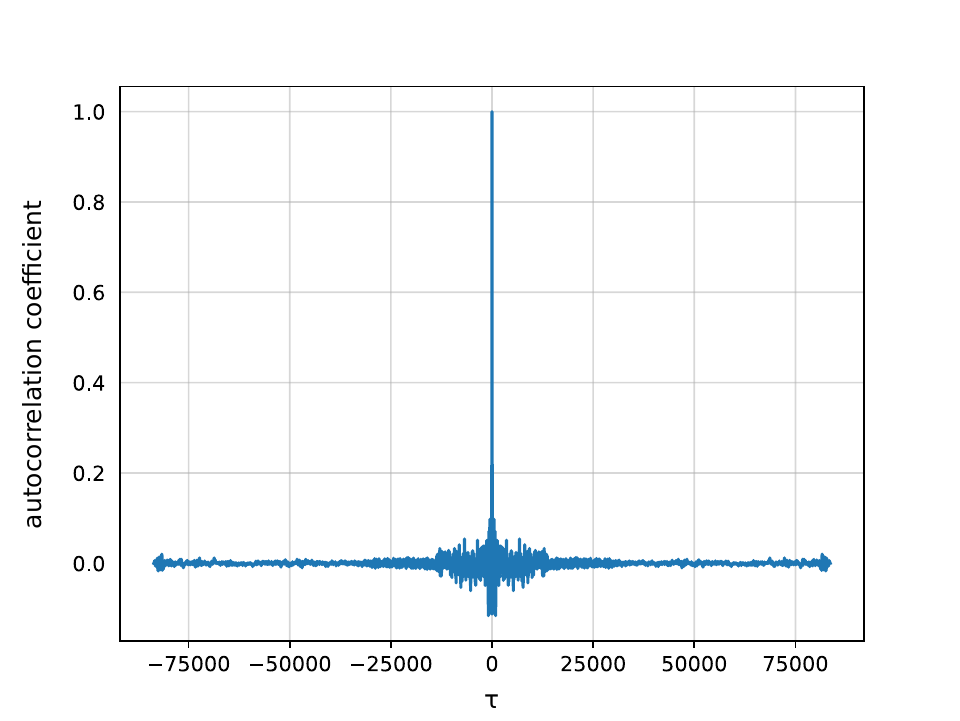}
    \caption{Autocorrelation coefficient of normalized local gradient, considering an FL task with CNN over CIFAR-10 dataset.}
    \label{fig_corr}
\end{figure}
A practical learning model may contain thousands or more parameters, implying that $D$ can be extremely large. Meanwhile, the normalized local gradient $\bar{g}_{m,t}[d]$ can be modeled as an independent zero-mean variable, leading to 
\begin{align}
    \sum_{d=1}^{D}\bar{g}_{m,t}[d]\bar{g}_{m,t}[d-\tau] \approx D\delta(\tau).\label{corr}
\end{align}
For instance, considering an FL task with convolutional neural network (CNN) over CIFAR-10 dataset, Fig. \ref{fig_corr} shows the autocorrelation coefficient of $\bar{g}_{m,t}[d]$ as the function of $\tau$. Fig. \ref{fig_corr} demonstrates the inherent autocorrelation property
of local gradients as shown in (\ref{corr}). 

In this regard, $\mathbf{x}_{1,m,t}$ serves as a promising candidate for passive wireless sensing. This is analogous to the case with
orthogonal frequency-division multiplexing (OFDM) waveform \cite{berger2010signal}. Accordingly, per FL round $t$, WD $m$ intends to analyze the echo signal $\tilde{\mathbf{x}}_{1,m,t}=[\tilde{x}_{1,m,t}[1],\cdots, \tilde{x}_{1,m,t}[d],\cdots,\tilde{x}_{1,m,t}[D]]$ reflected by the sensing target for object positioning, where 
\begin{align}
    \tilde{x}_{1,m,t}[d]\! =\! &\sum_{n\neq m}\!\underbrace{\!\left(\alpha_{m,n}x_{1,n,t}[d\!-\!\tau_{m,n}]\! +\! h_{m,n}x_{1,n,t}[d\!-\!\tau_{m,n}^{\mathrm{direct}}]\right)\!}_{\Theta_{m,n}[d](\text{inter-device interference})}  \notag \\ &+ \underbrace{\alpha_{m,m}x_{1,m,t}[d-\tau_{m,m}]}_{x_{1,m,t}^\mathrm{echo}[d] (\text{sensing echo})} + \omega^\mathrm{echo}_{m,t}.
\end{align} 
Specifically, $\tau_{m,n}= \left\lfloor\displaystyle{\frac{R_m}{cT_{slot}}+\frac{R_n}{cT_{slot}}}\right\rfloor$ is the number of slots during which the local gradient signal of WD $n$ travels to the WD $m$ via the target reflection. $\tau_{m,n}^{\mathrm{direct}}= \left\lfloor\displaystyle{\frac{R_{m,n}}{cT_{slot}}}\right\rfloor$ is the number of slots during which the local gradient signal of WD $n$ travels to the WD $m$ directly. $c$ is the speed of light. In addition, $\alpha_{m,n}$ is the radar target channel coefficient between WD $n$ and WD $m$. $h_{m,n}$ is the channel coefficient between WD $n$ and WD $m$, and $\omega^\mathrm{echo}_{m,t} \sim\mathcal{CN}(0,\sigma^2)$ is the AWGN.

\section{Distance Estimation and Target Localization}\label{sec_sensing}
In this section, we investigate zero-overhead sensing approaches by purely harnessing per-round FL local gradient signals. Specifically, a matched-filtering-based distance estimation method is proposed, based on which we develop a robust trilateration approach for target localization.

\subsection{Distance Estimation via Matched Filtering}\label{subsec_de}

To estimate the distance between the sensing target and WD $m$, we first apply the matched filter by correlating the conjugate of the known local gradient signal $\mathbf{x}_{1,m,t}$ with the echo signal $\tilde{\mathbf{x}}_{1,m,t}$. %\textcolor{black}{Compared with conventional over-the-air FL, the proposed sensing-native FL framework imposes no extra synchronization requirement for sensing, since each WD performs matched filtering with the echo of its own local gradient signal.} 
The output of the matched filter is
\textcolor{black}{
\begin{align}
    \Upsilon_{m,t}[i]\! &= \!\frac{1}{D}\sum_{d=1}^{D} x_{1,m,t}^*[d-i] \cdot \tilde{x}_{1,m,t}[d] \notag \\
    &= \frac{1}{D}\sum_{d=1}^{D}\! x_{1,m,t}^*[d\!-\!i]\!\!\left(\!\!x_{1,m,t}^\mathrm{echo}[d]\! + \!\! \sum_{n\neq m}\!\!\Theta_{m,n}\![d] \!+\! \omega^\mathrm{echo}_{m,t}\!\!\right)\! \notag \\
    &\overset{(a)}\approx \frac{1}{D}\sum_{d=1}^{D} x_{1,m,t}^*[d-i]\left(x_{1,m,t}^\mathrm{echo}[d] + \! \omega^\mathrm{echo}_{m,t}\right) \notag \\
    &\overset{(b)}\approx \frac{1}{D}\sum_{d=1}^{D} x_{1,m,t}^*[d-i] \cdot x_{1,m,t}^\mathrm{echo}[d],
\end{align}
}
\textcolor{black}{
where (a) is due to the statistical orthogonality among the normalized local gradient signals across different WDs; (b) follows from the fact that the cross-correlation term between the normalized local gradient signal and the zero-mean AWGN becomes negligible for sufficiently large $D$, i.e., $\frac{1}{D}\sum_{d=1}^{D}x_{1,m,t}^{*}[d-i]\omega_{m,t}^{\mathrm{echo}}\approx E[x_{1,m,t}^{*}[d-i]\omega_{m,t}^{\mathrm{echo}}]=E[x_{1,m,t}^{*}[d-i]]E[\omega_{m,t}^{\mathrm{echo}}]=0$. } 
%\textcolor{red}{ Here the actual mechanism of this term will vanish is that, the matched-filtering noise $\frac{1}{D}\sum_{d=1}^{D}x_{1,m,t}^{*}[d-i]\omega_{m,t}^{\mathrm{echo}}$ is also a random variable, but has zero mean and variance $\frac{\sigma^2}{D}|p_{1,m,t}|^2$, and hence the variance becomes negligible as $D$ grows large. So when D is sufficiently large, this random variable will converge to a constant $0$. So I am hesitant whether we should detailedly discribe this mechanism or just write as the blue one.}

These properties allow WD $m$ to detect the presence of its own local gradient component in the complicated echo signal, i.e.,
\textcolor{black}{
\begin{align}\label{mf}
    \Upsilon_{m,t}[i]
    &= \frac{\alpha_{m,m}|p_{1,m,t}|^2}{D} \sum_{d=1}^{D} \bar{g}_{m}[d-i] \bar{g}_{m}[d-\tau_{m,m}] \notag \\
    &\overset{(c)}\approx
    \begin{cases}
        \alpha_{m,m}|p_{1,m,t}|^2, & i = \tau_{m,m},\\
        0, & \text{otherwise},
    \end{cases}
\end{align}
}
where $(c)$ holds due to the autocorrelation property of the normalized local gradients in (\ref{corr}).

According to (\ref{mf}), the reflection path delay $\tau_{m,m}$ of WD $m$'s gradient signal at round $t$ is identified by
\begin{align}
    \ddot{\tau}_{m,m,t} = \arg\max_i |\Upsilon_{m,t}[i]|.
\end{align}
The estimated distance between WD $m$ and the sensing target is calculated as $\ddot{R}_{m,t} = \frac{c T_{\rm slot}}{2}\ddot{\tau}_{m,m,t}$.
Notice that the accuracy of the above per-round distance estimation approach relies on the i.i.d. assumption of local gradients across devices. In practice, to cope with scenarios where the across-device i.i.d. assumptions do not strictly hold, WD $m$ can record the historical estimation results in $\mathcal{R}_m^t = \{\ddot{R}_{m,1}, \ddot{R}_{m,2}, \cdots, \ddot{R}_{m,t}\}$ and select the majority to upload, i.e.,
\begin{align}\label{majority_select}
    \tilde{R}_{m,t} = \underset{r \in \mathcal{R}_m^t}{\arg\max} \sum_{i=1}^{t} \mathbb{I}\left( \ddot{R}_{m,i} = r \right),
\end{align}
where $\mathbb{I}(\cdot)$ is the indicator function. The refined distance $\tilde{R}_{m,t}$ is then uploaded to the BS via (\ref{tx_ud}).

\subsection{Target Localization}\label{rst-dc}

By receiving per-round distance estimate $\hat{R}_{m,t}$ from each WD via (\ref{r_hat}), the BS calculates the average distance between the sensing target and WD $m$, i.e., $\bar{R}_m = \frac{1}{T} \sum_{t=1}^{T} \hat{R}_{m,t}, \ \forall m$.

In addition to the across-device dependency of local gradients in practice, per-round distance error also results from the wireless channel fading and communication noise when transmitting $\tilde{R}_{m,t}$ to the BS, which further leads to inaccurate target localization. 

Typically, the BS leverages the trilateration method to determine the 3D coordinates of the object, which requires the estimated distances reported by at least three geographically-separated WDs. However, diverse locations and local datasets of WDs lead to heterogeneous errors of the distance estimations received at the BS. How to select these reported distance estimates for accurate object positioning is challenging.

In this paper, we propose a robust trilateration method with density-based clustering. Specifically, we first list all possible combinations of three different WDs, with the set of device indices for the $\gamma$-th combination defined as $\Lambda_\gamma$, $\gamma=1,2,\ldots,\Gamma$. The total number of combinations is $\Gamma = \binom{M}{3}$. For each combination $\Lambda_\gamma$, we determine the 3D position of the object $\mathbf{q}_\gamma = (\hat{x}_\gamma, \hat{y}_\gamma, \hat{z}_\gamma)$ by solving the following least-squares estimation problem:
\begin{align}\label{q_set}
    \mathbf{q}_\gamma = \arg\min_{\mathbf{q}\in \mathbb{R}^3} \sum_{m \in \Lambda_\gamma} \left(\| \mathbf{q} - \boldsymbol{\xi}_m \|_2  - \bar{R}_m \right)^2,
\end{align}
where $\boldsymbol{\xi}_m = (x_m, y_m, z_m)$ is the known position of WD $m$. Accordingly, a total number of $\Gamma$ candidates $\mathcal{Q} = \{ \mathbf{q}_\gamma \}_{\gamma=1}^\Gamma$ are generated.

Due to the heterogeneous distance estimation error of each WD, simply averaging $\mathbf{q}_\gamma$ from each device combination results in inaccurate localization performance. Instead, we adopt the density-based clustering (e.g., DBSCAN algorithm \cite{ester1996density}) for positioning candidates in $\mathcal{Q}$, producing $\kappa$ clusters $\{ \mathcal{C}_1, \ldots, \mathcal{C}_\kappa \}$ with poor-performance $\mathbf{q}_\gamma$ identified as outliers.
%\footnote{The performance of the DBSCAN-based clustering approach depends on two hyperparameters, i.e., the radius of the neighborhood $d_{neigh}$ and the minimum number of points $S_{pts}$ required to satisfy the core-point condition \cite{ester1996density}. In our simulations, we fix $S_{pts}$  and gradually increase $d_{neigh}$ until at least one cluster is identified.} 
We subsequently select the cluster $\mathcal{C}^* = \arg \max_{\kappa} |\mathcal{C}_\kappa|$ with the largest size, and compute its geometric centroid $\mathbf{\Xi} = (x^*, y^*, z^*)$ as the resulting sensing target position, i.e., 
\begin{align}
        \mathbf{\Xi} = \arg\min_{\mathbf{\Xi} \in \mathbb{R}^3} \sum_{\mathbf{q}_i \in \mathcal{C}^*} \|\mathbf{\Xi}-\mathbf{q}_i\|_2^2.
        \label{eq:centroid}
    \end{align}
\section{Learning Convergence Analysis and Problem Formulation}\label{sec_analysis}
In this section, we first analyze the convergence gap of the proposed sensing-native over-the-air FL, based on which the joint optimization of transmit power budget allocation and receiver beamforming is formulated to maximize the FL performance.

\subsection{Convergence Analysis}\label{sec:perf_analysis}
To facilitate the convergence analysis, we first introduce the following assumptions on the loss function $F(\cdot)$:
\begin{assumption}\label{L-smooth}
    $F$ is Lipschitz smooth with parameter $L$. That is,
\begin{align}
        F(\mathbf{w}) \!\leq\! F(\mathbf{w}') &\!+\! (\mathbf{w}-\mathbf{w}')^T \nabla F(\mathbf{w}^{\prime}) \!+\!\frac{L}{2} ||\mathbf{w}-\mathbf{w}^{\prime}||_2^2.
    \end{align}
\end{assumption}

\begin{assumption}\label{PL-inequality}
    $F$ satisfies the Polyak-Lojasiewicz (PL) inequality. Let $\textbf{w}^*=\arg\min_\textbf{w}F(\textbf{w})$. The loss $F$ is PL if, for all \textbf{w}, we have
\begin{align}
        ||\nabla F(\mathbf{w})||_2^2 \geq 2S(F(\mathbf{w})-F(\mathbf{w}^*)).
    \end{align}
\end{assumption}

\begin{assumption}\label{twice_diff}
    $F$ is twice-continuously differentiable.
\end{assumption}

\textcolor{black}{
%Notice that the above assumptions cover general non-convex loss functions in practical learning tasks with modern deep neural networks.
Notice that the above assumptions encompass general non-convex loss functions commonly encountered in practical learning tasks involving modern deep neural networks \cite{r1,allen2019learning,li2019convergence,stich2018local}. In particular, unlike strong convexity, the PL inequality \cite{polyak1964gradient} in \textit{\textbf{Assumption \ref{PL-inequality}}} does not require the loss function to be convex and permits the existence of multiple global minimizers, making it compatible with the highly non-convex loss landscapes of deep learning. Moreover, recent studies \cite{allen2019learning, karimi2016linear} have demonstrated that sufficiently overparameterized neural networks, which are prevalent in modern practice, provably satisfy the PL condition. }

%In particular, the PL inequality in Assumption 2 does not require strong convexity of the loss functions, and is satisfied by several practically relevant models, such as least-squares and logistic regression; Moreover, recent analyses of overparameterized neural networks suggest that the training trajectory may exhibit a benign local geometry, which makes the PL condition a meaningful analytical assumption.\cite{karimi2016linear}.

Recall that the global model update can be rewritten as $\mathbf{w}_{t+1} = \mathbf{w}_t - \lambda(\nabla F(\mathbf{w}_t)-\mathbf{e}_t)$,
where the per-round aggregation error $\mathbf{e}_t$ is given by
\begin{align}\label{et}
    \mathbf{e}_t &=  \frac {1}{\sum_{m=1}^MK_m}(\mathbf{r}_t- \hat {\mathbf{r}}_t).
\end{align}
We first characterize this error in the following lemma.

\begin{lemma}\label{Le_error_bound}
Under the transmit power constraints (\ref{pg_c}), (\ref{pv_c}), and (\ref{pu_c}), by setting the transmitter precoding factor and receiver normalizing scalar as\footnote{\textcolor{black}{For each model aggregation round in practice, each WD first uploads its local gradient variance to the BS via (\ref{tx_nu}). Then, the BS determines the transmitter precoding factor in (\ref{pg}) for each WD and
receiver normalizing scalar in (\ref{eta}). Upon receiving the feedback of $p_{1,m,t}$ from the BS, the WDs transmit their local gradients via (\ref{tx_grad}), which are leveraged not only for over-the-air model aggregation but also for zero-overhead distributed wireless sensing. Afterward, the locally estimated distance is reported through the statistics-bearing signal in (\ref{tx_ud}).}}
\begin{align}
        &p_{1,m,t}=\frac{K_m\sqrt{\eta_t}\hat{\nu}_{m,t}}{\mathbf{f}^H\mathbf{h}_{bs,m}},\quad \forall m, t,\label{pg} \\
        &\eta_t=\min_{m\in M}\frac{P_G|\mathbf{f}^H\mathbf{h}_{bs,m}|^2}{DK_m^2\hat{\nu}_{m,t}^2}, ~\forall t \label{eta}, 
    \end{align}
the expected norm of per-round aggregation error $||\mathbf{e}_{t}||_2^2$ is upper-bounded by
\begin{align}\label{L3_1}
        \mathbb{E}[||\mathbf{e}_{t}||_2^2] \!\leq\! E_1(P_\mu,P_\nu,\mathbf{f}) \!+\! E_2(P_G,\mathbf{f}) \!+\! E_3(P_\nu,P_G,\mathbf{f}),
    \end{align}
where 
\begin{align}
        &E_1(P_\mu,P_\nu,\mathbf{f}) \!=\! \frac{D\sigma^2} { K^2} \left(\! \frac{\mu_{\text{max}}^2}{P_\mu}+\frac{\nu_{\text{max}}^2}{P_\nu} \!\right)\!\sum_{m=1}^M \! \frac{K_m^2}{|\mathbf{f}^H\mathbf{h}_{bs,m}|^2}, \\
        &E_2(P_G,\mathbf{f}) = \frac{D^2\sigma^2\nu_{\text{max}}^2} {P_G K^2} \max_{m\in M}\frac{K_m^2}{|\mathbf{f}^H\mathbf{h}_{bs,m}|^2}, \\
        &E_3(P_\nu,P_G,\mathbf{f}) = \frac{D^2\sigma^4\nu_{\text{max}}^2} {P_GP_\nu K^2} \max_{m\in M} \frac{K_m^2}{|\mathbf{f}^H\mathbf{h}_{bs,m}|^4}.
    \end{align}
\end{lemma}
\begin{proof}
    See Appendix \ref{proof_L1}.
\end{proof}

From \textit{\textbf{Lemma \ref{Le_error_bound}}}, we observe that the per-round aggregation error comprises the following three distinct components. 

\begin{itemize}
    \item $E_1$ represents the aggregation error caused by noisy transmission of local gradient statistics (i.e., $\mu_{m,t}$ and $\nu_{m,t}$). We observe that better channel conditions or higher budgets allocated to $P_\mu$ and $P_\nu$ lead to a smaller $E_1$.
    \item $E_2$ stands for the over-the-air local-gradient aggregation error, which is dominated by the WD with the worst channel condition. This error can be improved by allocating a higher budget $P_G$ when transmitting local gradients. It is worth noting that with error-free transmission of local gradient statistics, per-round aggregation error is solely determined by $E_2$, which matches the prior analysis in \cite{yang2020federated,r1,cao2021optimized}.
    \item $E_3$ characterizes the coupling effect between the noisy statistics transmission and over-the-air gradient aggregation error. Specifically, the BS can only obtain the estimated variance of the local gradient when determining the optimal transmit scaling $p_{1,m,t}$ in (\ref{pg}) for per-round error minimization. In this regard, inaccurate variance estimation may amplify over-the-air gradient aggregation error, calling for the joint effort of optimizing budgets for $P_G$ and $P_\nu$.
\end{itemize}

Notice that, the transmit power budget allocated to gradient statistics is negatively correlated with that assigned to local gradient transmission due to the total transmit power constraint of each WD. Accordingly, increasing the gradient transmit power budget $P_G$ can lower over-the-air gradient aggregation error $E_2$, but may enlarge the errors $E_1$ and $E_3$ corresponding to the noisy transmission of statistic information due to a smaller power budget $P_\mu, P_\nu$ left. In addition, unlike conventional over-the-air FL transceiver design where the receiver beamforming vector is designed to improve the wireless channel condition of the straggler WD, the noisy transmission of gradient statistics further exacerbates the challenges of beamforming design since the error $E_1$ combines the wireless channel effects of all WDs.

Based on Lemma \ref{Le_error_bound}, we provide the convergence analysis for the proposed sensing-native over-the-air FL framework in the following theorem.

\begin{theorem}\label{Th_gap}
When \textit{\textbf{Assumptions 1-3}} hold and the learning rate $\lambda = \frac{1}{L}$, we have 
\begin{align}\label{theorm1}
        \mathbb{E}\!&\left[F(\mathbf{w}_{t})-F(\mathbf{w}^*)\right] \notag \\
        &\leq\!\Psi^t\!\left(F(\mathbf{w}_0)\!-\!F(\mathbf{w}^*)\right) + \frac {E_1 \!+\! E_2 \!+\! E_3 } {2L} \!\left( \frac{1-\Psi^t}{1-\Psi}\right),
    \end{align}
where the convergence speed $\Psi = 1 - \frac{S}{L}$.
\end{theorem}
\begin{proof}
    See Appendix \ref{proof_T1}.
\end{proof}

\textit{\textbf{Theorem \ref{Th_gap}}} indicates that the convergence rate is primarily determined by $\Psi$. As $t$ goes to infinity, the constant error floor is proportional to the aggregation error bound $E_1 + E_2 + E_3$, which quantitatively indicates the learning performance of the proposed sensing-native over-the-air FL framework.

\subsection{Problem Formulation}\label{sec:prob_formulation}
Based on the analysis above, we formulate the joint optimization of transmit power budget allocation and receiver beamforming to minimize the convergence error bound while satisfying the transmit power constraint of each WD, i.e., 
\begin{subequations}
\begin{align}
    \mathbf{P1}:\min_{\mathbf{f},P_G,P_{\mu},P_{\nu}}~
    &\frac{D\sigma^2} { K^2} \left( \frac{\mu_{\text{max}}^2}{P_\mu}+\frac{\nu_{\text{max}}^2}{P_\nu} \right)\!\sum_{m=1}^M  \frac{K_m^2}{|\mathbf{f}^H\mathbf{h}_{bs,m}|^2} \notag \\
    &+\frac{D^2\sigma^2\nu_{\text{max}}^2} {P_G K^2} \max_{m\in M}\frac{K_m^2}{|\mathbf{f}^H\mathbf{h}_{bs,m}|^2} \notag \\
    &+\frac{D^2\sigma^4\nu_{\text{max}}^2} {P_GP_\nu K^2} \max_{m\in M} \frac{K_m^2}{|\mathbf{f}^H\mathbf{h}_{bs,m}|^4}, \label{obj}\\
    {s.t.}\quad\quad
    &\text{C1}:||\mathbf{f}||_2^2 = 1, \\
    &\text{C2}:P_G+P_{\mu}+P_{\nu} \leq P_{\text{max}},  \\
    &\text{C3}:P_G \geq 0 , P_{\mu} \geq 0 , P_{\nu} \geq 0. 
\end{align}
\end{subequations}

Solving Problem ($\mathbf{P1}$) is challenging due to the non-convex nature and strong coupling between the beamforming vector $\mathbf{f}$ and the power budget allocation variables $\{P_G, P_\mu, P_\nu\}$. 

\section{Joint Transmit Power Budget Allocation and Beamforming Vector Optimization}\label{sec_opt}
In this section, we propose an efficient alternating optimization method to solve Problem ($\mathbf{P1}$). Specifically, under the fixed beamforming vector $\mathbf{f}$, we first derive the optimal transmit power budget allocation in closed forms. Accordingly, a successive convex approximation (SCA)-based approach is proposed to optimize the receiver beamforming scheme.

\subsection{Optimal Power Budget Allocation}\label{subsec:power_opt}
Given the receiver beamforming vector $\mathbf{f}$, we derive the optimal power budgets in the following two propositions.

\begin{proposition}\label{prop:P_mu_convex}
The optimal power budget $P_\mu^*$ allocated to local gradients mean is the unique solution of $\hat{J}(P_\mu)=0$, where 
\begin{align}\label{eq:derivative_expression}
    \hat{J}(P_\mu) \!=\! &- \!\frac{a}{P_\mu^2}\!+\!\frac{\left(\!\sqrt{\!c\!+\!b(\!P_{\text{max}}\!-\!P_\mu)} \!+\!\! \sqrt{\!c\!+\!d(\!P_{\text{max}}\!-\!P_\mu\!)}\right)^2}{(P_{\text{max}}-P_\mu)^3} \notag\\
    &\!\times \!\!\left(\! 1 \!+\! \frac{c}{\sqrt{[c\!+\!b(\!P_{\text{max}}\!\!-\!\!P_\mu)][c\!+\!d(P_{\text{max}}\!\!-\!\!P_\mu)]}} \right),
\end{align}
\begin{align}
    a \!&=\! \frac{D\sigma^2\mu_{\text{max}}^2} { K^2} \!\sum_{m=1}^M \!\frac{K_m^2}{|\mathbf{f}^H\mathbf{h}_{bs,m}|^2}, b \!= \!\frac{D\sigma^2\nu_{\text{max}}^2} { K^2} \!\sum_{m=1}^M \!\frac{K_m^2}{|\mathbf{f}^H\mathbf{h}_{bs,m}|^2}, \notag\\
    c &\!=\! \frac{D^2\!\sigma^4\nu_{\text{max}}^2} { K^2} \! \!\max_{m\in M}\! \frac{K_m^2}{|\mathbf{f}^H\!\mathbf{h}_{bs,m}\!|^4}\!, \notag
    d \!=\! \frac{D^2\sigma^2\nu_{\text{max}}^2}{ K^2}\!\!\max_{m\in M}\!\frac{K_m^2}{|\mathbf{f}^H\!\mathbf{h}_{bs,m}\!|^2}.
\end{align}
Here, $\hat{J}(P_\mu)$ is a monotonically increasing function for $P_\mu\in (0,P_{\text{max}})$. Since $\lim_{P_\mu \to 0^+} \hat{J}(P_\mu) < 0$ and $\lim_{P_\mu \to P_{\text{max}}^-} \hat{J}(P_\mu) > 0$, a unique solution $P_\mu^* \in (0, P_{\text{max}})$ exists for $\hat{J}(P_\mu)=0$.
\end{proposition}
\begin{proof}
    See Appendix \ref{app:proof_pmu}.
\end{proof}

\begin{proposition}\label{prop:P_g_p_nu}
With the optimal $P^*_\mu$, the optimal power budgets allocated to the local gradients and their variances are given by
\begin{align}
    P_{G}^* &=  \epsilon(P_{\text{max}} - P_\mu^*) , \label{eq:opt_PG_closed}\\
    P_{\nu}^* &= (1-\epsilon)(P_{\text{max}} - P_\mu^*),\label{eq:opt_Pnu_closed}
\end{align}
where
\begin{align}
    \epsilon=\frac{\sqrt{c+d(P_{\text{max}} - P_\mu^*)}} {\sqrt{c+d(P_{\text{max}} - P_\mu^*)}+\sqrt{c+b(P_{\text{max}} - P_\mu^*)}} \label{eq:scale}.
\end{align}
\end{proposition}

\begin{proof}
See Appendix \ref{app:proof_pnu_pg}.
\end{proof}

From the above propositions, we have the following key observations:
\begin{itemize}
    \item The scaling factor $\epsilon$ in (\ref{eq:scale}) is monotonically increasing with respect to coefficient $d$. That is, when the straggler suffers from a worse channel condition (i.e., a larger $d$), we intend to allocate a higher power budget $P_G$ for local gradient transmission. In addition, the power budget allocated to variance transmission increases with coefficient $b$. This is because a higher transmit power for the gradient statistics is required to combat the deep fading experienced by all WDs. 
    \item When the straggler WD experiences severe channel degradation (i.e., $\max_{m\in M} \frac{{\sigma^2}}{|\mathbf{f}^H\mathbf{h}_{bs,m}|^2} \gg 1$), the coefficient $c$ dominates the power budget allocation (i.e., $c \gg b$ and $c \gg d$). Consequently, the power allocation asymptotically converges to an equal assignment strategy, i.e., $P_G^* \approx P_\nu^* \approx \frac{1}{2}(P_{\text{max}} - P_\mu^*)$. This demonstrates that the transmission of gradient statistics is as critical as the local gradients themselves. Inaccurate statistics transmission affects both wireless sensing and FL performance for the proposed sensing-native over-the-air FL framework.
\end{itemize}

Based on \textit{\textbf{Proposition \ref{prop:P_mu_convex}}}, one can efficiently find the optimal $P_\mu^*$ using bisection search method, based on which the optimal $P_G^*$ and $P_\nu^*$ can be directly obtained via the closed-form expressions in \textit{\textbf{Proposition \ref{prop:P_g_p_nu}}}.

\subsection{Receive Beamforming Vector Optimization}\label{subsec:beamforming_opt}
Given the optimal power budget allocation, the optimization problem over the receiver beamforming is 
\begin{subequations}
\begin{align}
\mathbf{P2}:\min_{\mathbf{f}}~&\left(a'+b'\right)\sum_{m=1}^{M} \frac{K_m^2}{|\mathbf{f}^H\mathbf{h}_{bs,m}|^2}\notag \\ 
&+c'\max_{m \in M} \frac{K_m^2}{|\mathbf{f}^H\mathbf{h}_{bs,m}|^4}+d'\max_{m \in M} \frac{K_m^2}{|\mathbf{f}^H\mathbf{h}_{bs,m}|^2}, \ \ \notag \label{P4} \\
    {s.t.}\quad &||\mathbf{f}||_2^2 = 1, \notag
\end{align}
\end{subequations}
where 
    $a'= \frac{D\sigma^2\mu_{\text{max}}^2} {P_\mu K^2},b'= \frac{D\sigma^2\nu_{\text{max}}^2} {P_\nu K^2},c'= \frac{D^2\sigma^4\nu_{\text{max}}^2} {P_\nu P_G K^2},d'= \frac{D^2\sigma^2\nu_{\text{max}}^2}{P_G K^2}$.

\begin{proposition}
    The feasible region of Problem (\textbf{P2}) can be convexified to the closed unit ball $\|\mathbf{f}\|_2^2 \leq 1$, which preserves the optimality.
\end{proposition}
\begin{proof}
We prove this proposition by contradiction. Suppose that $\mathbf{f}_1^*$ is optimal for $\mathbf{P2}$ with $\|\mathbf{f}_1^*\|_2^2 < 1$. Construct $\mathbf{f}_2 \triangleq \mathbf{f}_1^* / \|\mathbf{f}_1^*\|_2$, which satisfies $\|\mathbf{f}_2\|_2^2 = 1$. Then, we have
\begin{align}
    \left|\mathbf{f}_2^H \mathbf{h}_{bs,m}\right|^2 = \frac{|\mathbf{f}_1^{*H} \mathbf{h}_{bs,m}|^2}{\|\mathbf{f}_1^*\|_2^2} > |\mathbf{f}_1^{*H} \mathbf{h}_{bs,m}|^2.
\end{align}
Given that the objective function is monotonically decreasing in $|\mathbf{f}^H \mathbf{h}_{bs,m}|^2$, $\mathbf{f}_2$ yields a lower objective than $\mathbf{f}_1^*$. This contradicts the optimality assumption of $\mathbf{f}_1^*$. Therefore, the non-convex constraint $\|\mathbf{f}\|_2^2 = 1$ in the original Problem ($\mathbf{P2}$) can be converted to $\|\mathbf{f}\|_2^2 \leq 1$ without loss of optimality.
\end{proof}

Then, by defining $ \Phi_m(\mathbf{f}) \triangleq  |\mathbf{f}^H\mathbf{h}_{bs,m}|^2$, $ u_1 \triangleq \min \frac{\Phi_m^2}{K_m^2}  $ and $ u_2 \triangleq \min \frac{\Phi_m}{K_m^2} $, Problem ($\mathbf{P2}$) can be equivalently expressed as
\begin{subequations}
\begin{align}
\mathbf{P3}:\min_{\mathbf{f}, u_1, u_2}~&\left(a'+b'\right)\sum \frac{K_m^2}{\Phi_m} +\frac{c'}{u_1}+\frac{d'}{u_2}, \label{P5} \\
    {s.t.}\quad &\text{C6}:||\mathbf{f}||_2^2 \leq 1, \\
    &\text{C7}:K_m^2u_1\leq \Phi_m^2, \ \ \forall m, \\
    &\text{C8}:K_m^2u_2\leq \Phi_m, \ \ \forall m.
\end{align}
\end{subequations}

Let $A$, $\{B_m\}$, and $\{C_m\}$ be the dual variables associated with constraints C6, C7, and C8, respectively. Then the Lagrangian of Problem ($\mathbf{P3}$) is given by
\begin{align}\label{lagrangian}
\mathcal{L} = &\left(a'+b'\right)\sum_{m=1}^{M} \frac{K_m^2}{\Phi_m} +\frac{c'}{u_1}+\frac{d'}{u_2} + A(||\mathbf{f}||_2^2 - 1) \notag \\
& \!+ \!\sum_{m=1}^{M} \!B_m (K_m^2 u_1 - \Phi_m^2) \!+\! \sum_{m=1}^{M}\! C_m (K_m^2 u_2 - \Phi_m).
\end{align}
\begin{proposition}\label{beamforming}
By defining $\mathbf{H}_{bs,m} \triangleq \mathbf{h}_{bs,m} \mathbf{h}_{bs,m}^H$, and a weighted channel covariance matrix 
$\mathbf{G} \triangleq \left(\sum_{m=1}^{M} \left[ \frac{(a'+b')K_m^2}{\Phi_m^2(\mathbf{f}^*)} + 2B_m^ *\Phi_m(\mathbf{f}^*) + C_m^* \right] \mathbf{H}_{bs,m}\right)$, the optimal beamforming vector $\mathbf{f}^*$ is the principal eigenvector satisfying $\mathbf{G} \mathbf{f}^* = A^* \mathbf{f}^*$.
\begin{proof}
See Appendix \ref{proof_dual}.
\end{proof}
\end{proposition}

{Notice that only these constraints in C7 and C8 corresponding to the straggler with the worst channel condition are active (i.e., hold with equality).} According to the KKT complementary slackness theorem, only the dual variables $B_m$ and $C_m$ associated with the straggler $m$ are non-zero. These non-zero dual variables significantly increase the contribution of the stragglers' channel covariance matrices $\mathbf{H}_{bs,m}$ to the matrix $\mathbf{G}$. As a result, the optimal receiver beamforming is preferentially aligned with the straggler’s channel.

In addition to the insights observed from \textit{\textbf{Proposition \ref{beamforming}}}, we further propose an efficient SCA-based approach to find the optimal receiver beamforming scheme. Specifically, for iteration $i=1\dots I_{\text{max}}$, we approximate $\Phi_m$ by the first-order Taylor expansion based on the current value $\mathbf{f}^{(i)}$, i.e., 
\begin{align}
    {\Phi}_m^{(i)} &= \Phi_m(\mathbf{f}^{(i)}) + \text{Re}\left\{ (\mathbf{f}-\mathbf{f}^{(i)})^H \nabla \Phi_m(\mathbf{f}^{(i)}) \right\} \notag \\
    &= \Phi_m(\mathbf{f}^{(i)}) + 2\text{Re}\left\{ (\mathbf{f} - \mathbf{f}^{(i)})^H  \mathbf{H}_{bs,m} \mathbf{f}^{(i)}  \right\}\label{taylor1}.
\end{align}
Similarly, the terms $\frac{1}{\Phi_m}$ and $\Phi_m^2$ are approximated by
\begin{align}
    &\left(\frac{1}{{\Phi}_m}\right)^{(i)} \!= \!\frac{1}{\Phi_m(\mathbf{f}^{(i)})} \!-\! \frac{2\text{Re}\!\left\{\! (\mathbf{f} - \mathbf{f}^{(i)})^H  \mathbf{H}_{bs,m} \mathbf{f}^{(i)} \!\right\}}{\Phi_m^2(\mathbf{f}^{(i)})},\label{taylor2}\\
    &\left(\!\Phi_m^2\!\right)^{\!(i)} \!\!=\!\Phi_m^2(\mathbf{f}^{(i)}) \!+\! 4\Phi_m(\!\mathbf{f}^{(i)})\text{Re}\!\!\left\{\! (\mathbf{f} \!\!-\! \mathbf{f}^{(i)})^{\!H}  \!\mathbf{H}_{bs,m} \mathbf{f}^{(i)}\! \!\right\}\label{taylor3}.
\end{align}
Leveraging (\ref{taylor1})-(\ref{taylor3}), we construct the following second-order cone problems (SOCPs)
\begin{subequations}
\begin{align}
\mathbf{P4}:\mathbf{f}^{(i+1)}\!=&\arg\!\min_{}\!\left(a'+b'\right)\!\sum K_m^2 \!\left(\frac{1}{{\Phi}_m}\right)^{(i)}\!\!+\!\frac{c'}{u_1}\!+\!\frac{d'}{u_2}, \label{P6} \notag\\
    {s.t.}\quad &\text{C6}:||\mathbf{f}||_2^2 \leq 1, \notag\\
    &\text{C7}:K_m^2u_1\leq \left({{\Phi}_m^2}\right)^{(i)}, \ \ \forall m,  \notag\\
    &\text{C8}:K_m^2u_2\leq {\Phi}_m^{(i)}, \ \ \forall m, \notag
\end{align}
\end{subequations}
which can be efficiently solved via off-the-shelf solvers, such as CVX \cite{grant2008cvx}. We summarize the overall algorithm to solve Problem ($\mathbf{P1}$) in \textit{\textbf{Algorithm \ref{alg:joint_opt}}}.
\begin{algorithm}[t]
\caption{The Proposed Algorithm for Problem (\textbf{P1})}
\label{alg:joint_opt}
\begin{algorithmic}[1]
\State \textbf{Input:} $\{\mathbf{h}_{bs,m},K_m\}, D, \sigma, M,P_{\text{max}},It_{\text{max}}, I_{\text{max}},\varepsilon$.

\State \textbf{Initialize:} $\mathbf{f}$.

\For{$iter=1\dots It_{\text{max}}$}
    \State Update $P_\mu^{*}$ via \textit{\textbf{Proposition \ref{prop:P_mu_convex}}};
    \State Update $P_G^{*}$ and $P_\nu^{*}$ via \textit{\textbf{Proposition \ref{prop:P_g_p_nu}}};
    \For{$i=1\dots I_{\text{max}}$}
        \State Update $\mathbf{f}^{(i)}$ by solving ($\mathbf{P4}$);
    \EndFor
    \State Compute $\text{obj}^{(iter+1)}$ by (\ref{obj});
    \If{$|\text{obj}^{(iter+1)}-\text{obj}^{(iter)} |\leq \varepsilon$} \textbf{early stop};
    \EndIf
    \EndFor

\State \textbf{Return} $P_\mu^{*}, P_G^{*}, P_\nu^{*}, \mathbf{f}$.
\end{algorithmic}
\end{algorithm}

\section{Experiment Results}\label{sec_exp}
In this section, numerical evaluations are presented to validate the performance of the proposed sensing-native over-the-air FL framework. 
%The experiments are implemented in Python 3.8.10 with PyTorch 2.4.1, executing on a Linux server equipped with NVIDIA GeForce RTX 4090 GPUs (24GB VRAM) and CUDA 12.1 acceleration. 
%The complete source code repository is publicly available at \url{https://github.com/sukiiiKotori/over-the-air FL}.

\subsection{Experiment Setups}

We consider a single-cell mmWave wireless network within a $100 \times 100 \text{ m}^2$ deployment area. The BS is equipped with $N=32$ antennas positioned at Cartesian coordinates $(0,0,10)$. Unless otherwise stated, we consider $M=20$ single-antenna WDs randomly deployed in this area, with their heights following a uniform distribution $\mathcal{U}(0,2)$. The carrier frequency $f_c$ is 26 GHz. The channel noise power $\sigma^2$ is set to $10^{-10}\ \text{W}$. The bandwidth is set to $160 \ \text{MHz}$ and thus the symbol duration $T_{slot}$ is $\frac{1}{160\ \text{M}} = 6.25\ \text{ns}$.

The channel between the BS and WD $m$ follows the geometric model with $\rho=8$ paths \cite{gao2019spatial}, where $\mathbf{h}_{bs,m} = \text{PL}_{bs,m} \cdot \sqrt{\frac{N}{\rho_{\text{max}}}} \sum_{\rho=1}^{\rho_{\text{max}}} \alpha_{\rho} \boldsymbol{a}_{r}(\theta_{\rho})$. Here, $\text{PL}_{bs,m} = \sqrt{ c^{2} / (4\pi f_c R_{bs,m})^{2} }$ denotes the path loss, $\alpha_{\rho} \sim \mathcal{CN}(0,1)$ represents the complex gain of the $\rho$-th path, and $\theta_{\rho} \sim \mathcal{U}(0,2\pi)$ is the uniformly-distributed angle of arrival (AoA). The receive array response $\boldsymbol{a}_{r}(\theta)$ for half-wavelength spaced uniform linear array (ULA) takes the form $\boldsymbol{a}_r(\theta) = \frac{1}{\sqrt{N}} [ 1,\ e^{j\pi \sin\theta},\ \cdots,\ e^{j(N-1)\pi \sin\theta} ]^T$. 

The radar target channel gain follows $\alpha_{m,n} = \sqrt{ \sigma_{\mathrm{rcs}} G_{\mathrm{tx}} G_{\mathrm{rx}} c^{2} / [ (4\pi)^{3} f_c^2 (R_{m}R_{n})^{2} ] }$, where $\sigma_{\mathrm{rcs}}$ denotes the Radar Cross Section (RCS) \cite{ozkaptan2021optimal} and $G_{\mathrm{tx}} = G_{\mathrm{rx}}$ stand for the transmit and receiver antenna gains. For interference channel between WDs $m$ and $n$, we model $h_{m,n} = A_{\mathrm{side}}  \text{PL}_{m,n}$ with $A_{\mathrm{side}} = G_{\mathrm{tx}} - 15$ dB accounting for side-lobe attenuation, as WDs' antennas are not oriented toward the ground plane in which devices are located.

For the learning task, we train a CNN with $D = 83\text{,}594$ parameters on the CIFAR-10 dataset \cite{krizhevsky2009learning}. The network architecture consists of three convolutional blocks followed by three fully connected layers. Specifically, the convolutional layers utilize 16, 32, and 64 filters with kernel sizes of $5\times5$, $5\times5$, and $3\times3$, respectively. Each convolutional layer is succeeded by batch normalization, ReLU activation, and a $2\times2$ max-pooling layer. The resulting feature maps are flattened and processed by dense layers with 256 and 128 units, leading to a final 10-class output layer. The local training data are i.i.d. drawn from $K =50\text{,}000$ images with $K_m = K/M$. The learning rates for batch gradient descent (GD) and mini-batch stochastic gradient descent (SGD) are set to 0.03 and 0.1, respectively. The average performance of 30 independent evaluations is reported.

\subsection{Learning Performance Comparison}\label{exp-FL}
We compare the learning performance of the proposed sensing-native over-the-air FL framework with the following benchmarks:
\begin{figure*}[t]\centering
    \centering
    
    \begin{minipage}{0.31\textwidth}
        \includegraphics[width=\linewidth]{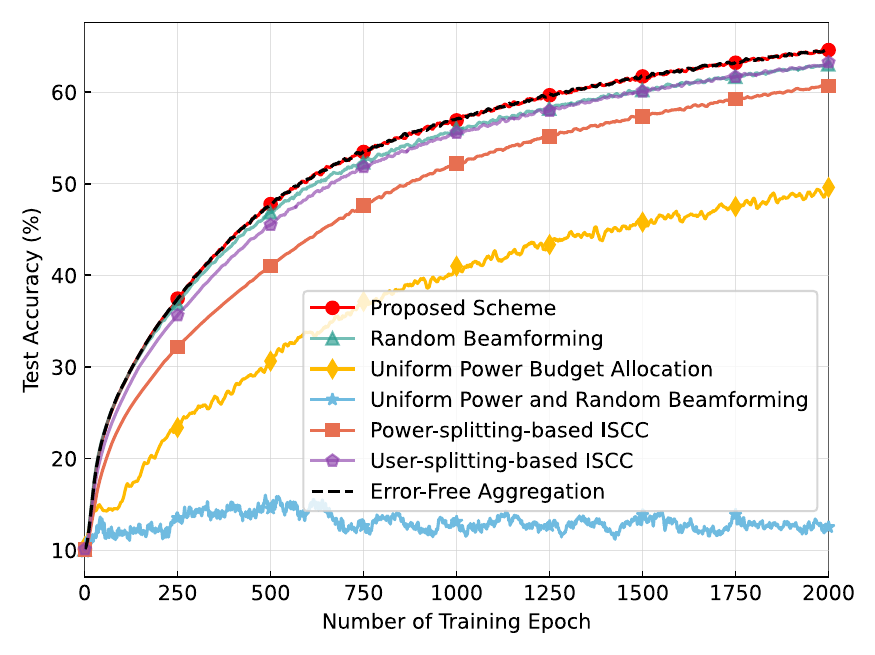}
        \caption{Comparison of test accuracy with 10 dBm per-slot average transmit power.}
        \label{fb_10db}
    \end{minipage}
    \hspace{0.02\textwidth}
    \begin{minipage}{0.31\textwidth}
        \includegraphics[width=\linewidth]{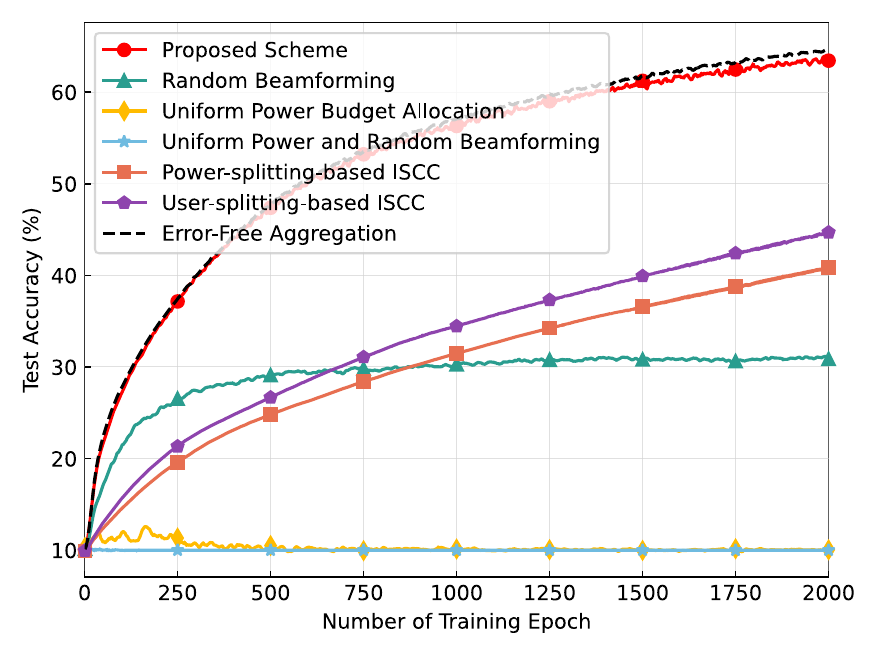}
        \caption{Comparison of test accuracy with -10 dBm per-slot average transmit power.}
        \label{fb_-10db}
    \end{minipage}
    \hspace{0.02\textwidth}
    \begin{minipage}{0.31\textwidth}
        \includegraphics[width=\linewidth]{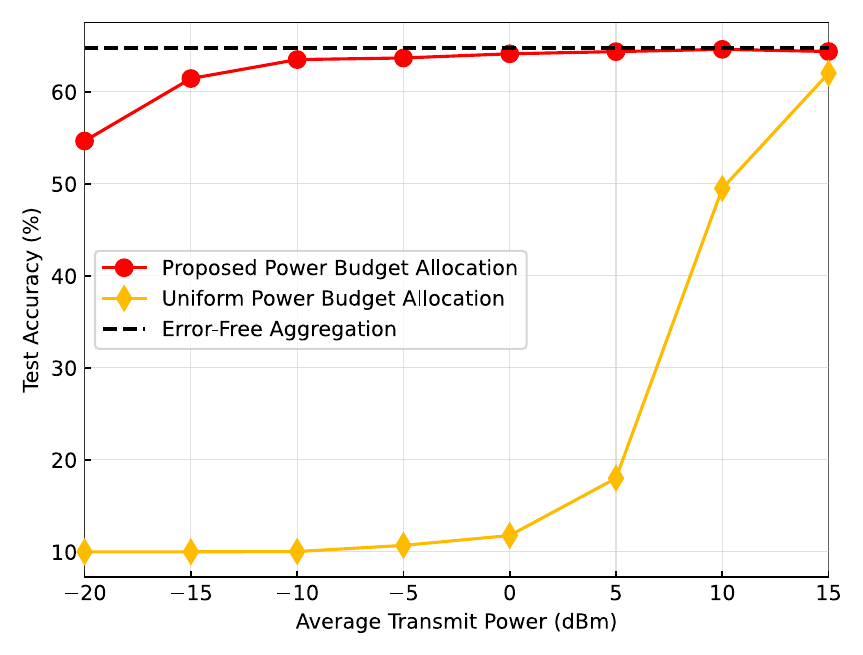}
        \caption{Test accuracy versus per-slot average transmit power under different power budget allocation schemes.}
        \label{fb_txp}
    \end{minipage}
\end{figure*}
\begin{itemize}
    \item \textbf{Error-free aggregation} \cite{mcmahan2017communication}: Per-round model aggregation is performed without channel fading and communication noise.
    
    \item \textbf{Random beamforming}: The transmit power budget is optimally allocated via \textit{\textbf{Proposition \ref{prop:P_mu_convex}}} and \textit{\textbf{ \ref{prop:P_g_p_nu}}} while setting the beamforming vector randomly.
    
    \item \textbf{Uniform power budget allocation} \cite{zhu2019broadband,r1}: The beamforming vector is optimized via the proposed SCA-based algorithm, while equally allocating the power budget for local gradient statistics transmission, i.e., $P_\mu = P_\nu = P_G/D$.
    
    \item \textbf{Uniform power budget allocation with random beamforming}: We follow equal power budget allocation \textit{and} randomly select receiver beamforming vector.

    \item \textcolor{black}{\textbf{Power-splitting-based ISCC} \cite{liu2024multi}: Each WD equally allocates its transmit power to the local gradient signal and a dedicated sensing waveform, where we assume noise-free transmission of the learning statistics and over-the-air transceiver design variables are optimized by minimizing  per-round aggregation MSE.}

    \item \textcolor{black}{\textbf{User-splitting-based ISCC} \cite{10615482}: The WDs are divided into two equal-sized groups. The WDs in one group transmit dedicated sensing waveforms with full transmit power. The other WDs transmit local gradient signals that serve a dual FL-and-sensing function, where over-the-air transceiver design variables are optimized by minimizing  per-round aggregation MSE under the assumption of noise-free transmission of the learning statistics. }
\end{itemize}

\textcolor{black}{In Fig. \ref{fb_10db}, we plot the test accuracies of all competing methods as the training proceeds, where the average transmit power per slot is 10 dBm. We observe that the proposed method approaches the performance with error-free aggregation and outperforms the other baselines, namely, random beamforming, uniform power budget allocation, uniform power budget allocation with random beamforming by around 1.6\%, 15.0\% and 52.1\%  after 2000 epochs. This demonstrates the effectiveness of jointly optimizing receiver beamforming and
the transmit power budget for both local gradients and the
corresponding statistics (i.e., gradient mean and variance) in
order to improve over-the-air FL performance while achieving
zero-overhead wireless sensing.
In addition, the proposed method also achieves higher test accuracy than the power-splitting-based ISCC and user-splitting-based ISCC methods by 3.9\% and 1.4\%, respectively. 
This is because both ISCC baselines require dedicating resources to the sensing module (i.e., transmit power in  power-splitting-based method and a subset of WDs in  user-splitting-based method), which degrades learning performance due to resource competition and cross-module interference. This result highlights the benefit of the proposed sensing-native over-the-air FL design, which unleashes the full ISCC resources for both sensing and FL without mutual compromise.}
%Power-splitting-based ISCC reduces the transmit power allocated to the FL signal and introduces sensing waveform interference, whereas User-splitting-based ISCC reduces the amount of data participating in FL and still introduces the interference from sensing-only WDs. 
%Moreover, User-splitting-based ISCC performs better than Power-splitting-based ISCC, probably because each WD holds a relatively large local dataset (2500 samples), making communication distortion more dominant than the data reduction.
%These observations demonstrate that the proposed sensing-native design avoids the power and user resource competition as well as the sensing-waveform interference suffered by Power-splitting-based ISCC and User-splitting-based ISCC, while the 

Fig. \ref{fb_-10db} compares the performance in terms of the test accuracy when average transmit power per slot is -10 dBm. Under such low-SNR condition, our proposed approach still achieves learning performance close to the error-free aggregation with only a 1.2\% reduction in test accuracy.
In addition, uniform power budget allocation scheme makes the model fail to converge. Moreover, the random beamforming scheme leads to around 32.5\% accuracy degradation, where the corresponding performance gap is larger compared with the relatively high-SNR case shown in Fig. \ref{fb_10db}. This highlights the importance of joint beamforming and power budget optimization for both local gradient and the corresponding statistics transmissions under the proposed sensing-native over-the-air FL system, especially in resource-constrained environments.
\textcolor{black}{
%Fig. \ref{fb_-10db} compares the performance in terms of the test accuracy when average transmit power per slot is -10 dBm. Under such low-SNR condition, our proposed approach still achieves learning performance close to the error-free aggregation with only a 1.2\% reduction in test accuracy. In addition, the uniform power budget allocation scheme makes the model fail to converge. Moreover, the random beamforming scheme leads to around 32.5\% accuracy degradation, where the corresponding performance gap is larger compared with the relatively high-SNR case shown in Fig. \ref{fb_10db}. 
The performance gaps between the proposed method and power-splitting-based ISCC and user-splitting-based ISCC are also  enlarged. This shows the advantages of the proposed sensing-native design by avoiding resource competition under low-SNR regime.}

Fig. \ref{fb_txp} depicts the final test accuracy versus the average per-slot transmit power after 2000 training rounds. We observe that a larger per-slot total transmit power budget leads to higher test accuracy for both optimal and uniform budget allocation schemes. In particular, the proposed method attains near-to-error-free performance around 0 dBm per-slot total power budget. 
In addition, the performance gap between the proposed transmit power budget allocation and the uniform budget assignment strategy becomes more pronounced in the low-SNR regime, showcasing the necessity of taking the transmit power allocation for local gradient statistics into account in the proposed sensing-native over-the-air FL. 

\begin{figure*}[t]\centering
    \centering
    \begin{minipage}{0.31\textwidth}
        \includegraphics[width=\linewidth]{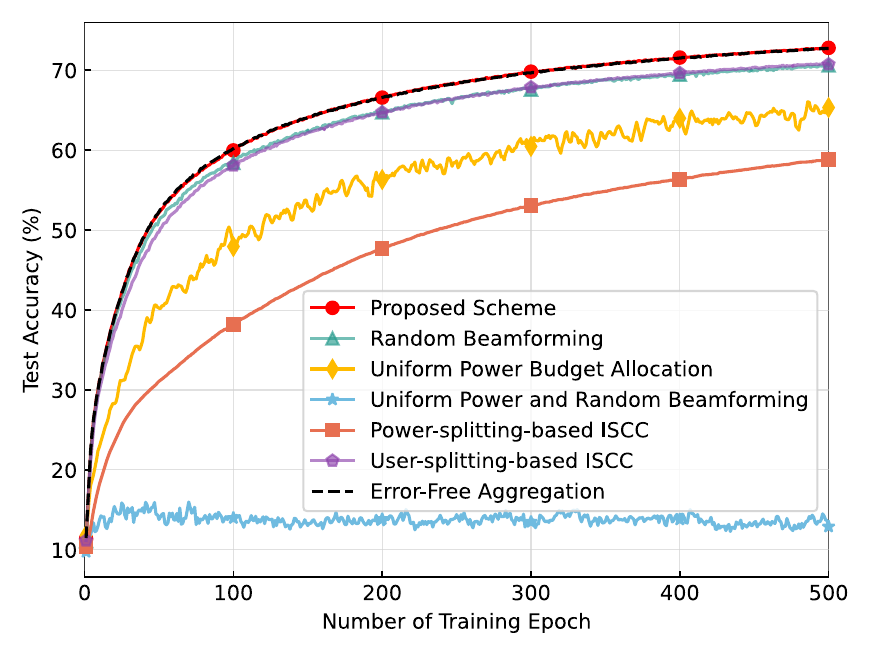}
        \caption{Comparison of test accuracy with 10 dBm per-slot average transmit power under mini-batch SGD.}
        \label{mb_10db}
    \end{minipage}
    \hspace{0.02\textwidth}
    \begin{minipage}{0.31\textwidth}
        \includegraphics[width=\linewidth]{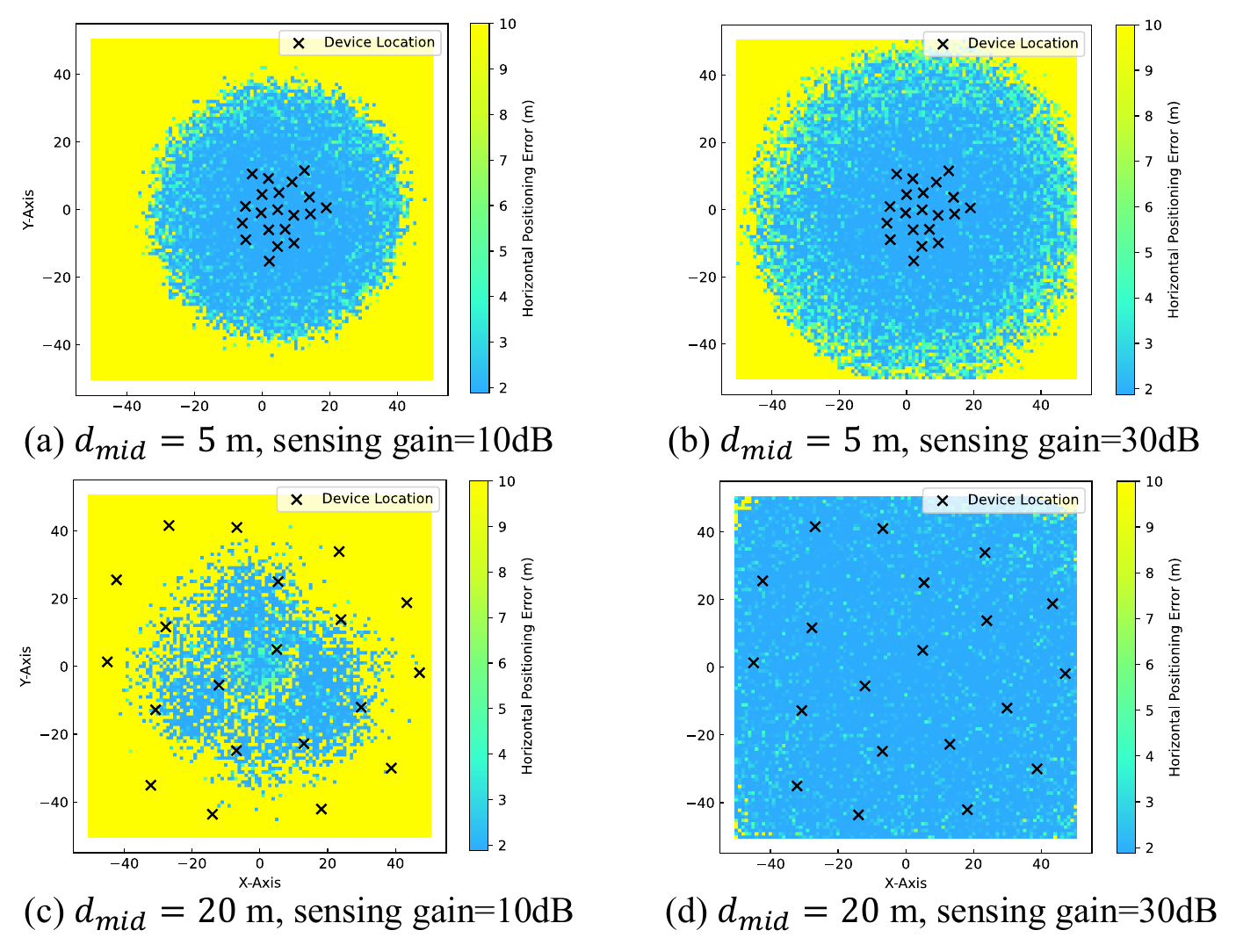}
        \caption{Positioning error distributions with different minimum inter-WD distance $d_{mid}$ and sensing gains.}
        \label{fig:error_distributions}
    \end{minipage}
    \hspace{0.02\textwidth}
    \begin{minipage}{0.31\textwidth}
        \includegraphics[width=\linewidth]{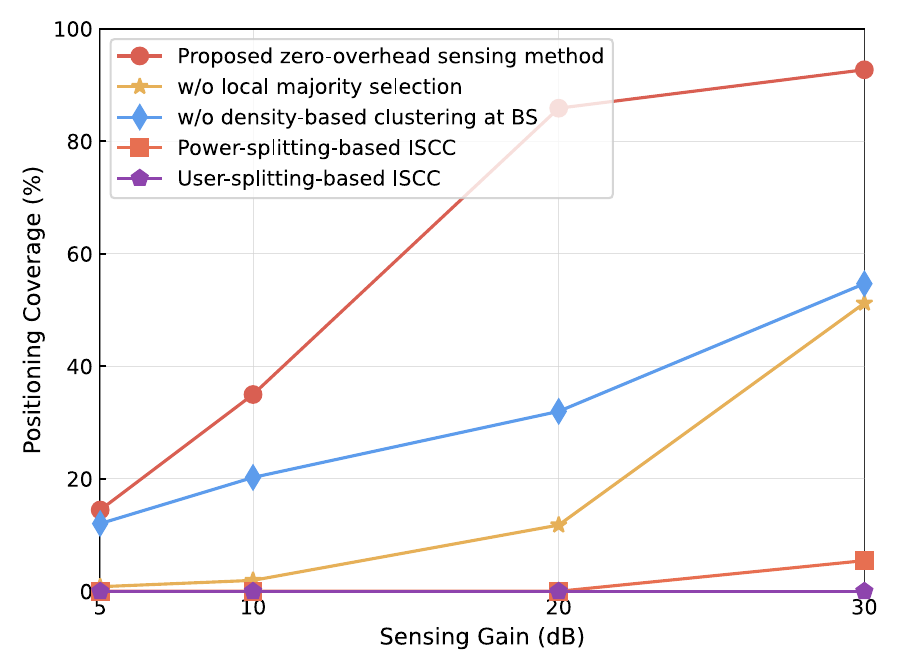}
        \caption{Positioning coverage comparison under different sensing gains.}
        \label{positioning_scheme}
    \end{minipage}
\end{figure*}

In Fig. \ref{mb_10db}, we further evaluate the learning performance of the proposed sensing-native over-the-air FL under mini-batch SGD, where each WD performs five local SGD iterations before each model aggregation step with mini-batch size $B_m = K_m/5$, and per-slot total power budget is set to 10 dBm. It is observed that the proposed scheme outperforms all baselines while closely approaching the performance of error-free aggregation method, confirming its effectiveness for mini-batch SGD framework. Notably, we observe a reduced performance gap between the proposed scheme and uniform power allocation compared to the full-batch GD framework. This is because multiple local iterations effectively decrease the frequency of over-the-air model aggregation perturbed by wireless channel fading and communication noise.

\subsection{Performance Evaluation of Zero-Overhead Sensing }\label{sensing_exp}
To evaluate the proposed zero-overhead distributed wireless sensing performance along with the over-the-air FL system, we discretize the considered $100\times 100 \text{ m}^2$ deployment area into $101 \times 101$ grids comprising 10,201 points. Suppose that sensing target is located at one of these grid points with 20 m height. We set per-slot total power budget as 10 dBm and adopt the mini-batch SGD framework for over-the-air FL. The total sensing gain is defined as $\sigma_{\mathrm{rcs}}G_{\mathrm{tx}}G_{\mathrm{rx}}$. A higher sensing gain represents a better radar target channel condition.

%\begin{figure}[t]\centering
%\includegraphics[width=0.435\textwidth]{positioning_error_cdf.pdf}
%\caption{Positioning error distribution of x, y, and z axes.}
%\label{error_distribution}
%\end{figure}

We consider the horizontal positioning error as the sensing performance metric, defined as $e(\mathbf{\Xi}) = \sqrt{(\mathbf{\Xi}[1] - \mathbf{\Xi}_{\text{true}}[1])^2 + (\mathbf{\Xi}[2] - \mathbf{\Xi}_{\text{true}}[2])^2}$. Then, the corresponding positioning coverage is defined as the proportion of grid points whose positioning error $e \leq 3 \text{m}$. 

Fig. \ref{fig:error_distributions} reports the positioning error at each grid point under different sensing gains and distribution densities among 20 WDs. {Note that the device distribution density is characterized by the minimum inter-WD distance (defined by $d_{mid}$). Specifically, WDs are deployed sequentially. The first WD is randomly placed on a circle of radius $d_{mid}$ centered at the BS, and each subsequent WD is positioned such that its distance to the nearest previously deployed WD equals $d_{mid}$. As a result, a smaller $d_{mid}$ indicates a more tightly clustered deployment.} We observe that the positioning coverage increases with sensing gain under each $d_{mid}$. This is because higher sensing gain enhances the echo SNR per WD, extending the effective sensing coverage. 
By comparing Figs. 7(a) and 7(b), we observe that with a closer geographic proximity among WDs (i.e., $d_{mid}$ = 5 m), the sensing coverage enhancement by enlarging sensing gain from 10 dB to 30 dB is limited (i.e., an increased sensing coverage from 37.2\% to 54.1\%). This is because although increasing sensing gain improves the radar target channel between any two WDs, the tight WD clustering causes severe inter-device interference, limiting the sensing coverage improvement. In contrast, we observe from Figs. 7(c) and 7(d) that when the WDs are deployed with a larger $d_{mid}$ (e.g., 20 m), the proposed sensing-native over-the-air FL system benefits significantly from increased sensing gain (i.e., an increased sensing coverage from 20.1\% to 94.6\%). This is due to the fact that a sparser deployment of WDs not only reduces inter-device sensing interference, but also enhances horizontal geometric diversity. A wide angular separation relative to the target reduces the Horizontal Dilution of Precision (HDOP) and improves the robustness against inter-device interference and communication noise.

\begin{table}[t]
    \centering
    \caption{Positioning Coverage (\%) vs. Number of Devices across Different Sensing Gains}
    \label{positioning_usrnum}
    \begin{tabular}{lccc}
        \toprule
        \multirow{2}{*}{\textbf{Total Sensing Gain (dB)}} & \multicolumn{3}{c}{\textbf{Number of Devices}} \\
        \cmidrule(lr){2-4}
        & \textbf{10} & \textbf{20} & \textbf{30} \\
        \midrule
        5  & 16.3 & 14.4 & 7.9  \\
        20 & 81.5 & 85.9 & 87.0 \\
        \bottomrule
    \end{tabular}
\end{table}

Table \ref{positioning_usrnum} compares the positioning coverage versus the number of WDs in the proposed sensing-native over-the-air FL system, where the minimum inter-WD distance $d_{mid}$ is set to 15 m. For the low sensing gain case (i.e., 5 dB), the system with 10 WDs outperforms those with 20 or 30 WDs. 
This is because with a fixed $d_{mid}$, a smaller number of WDs results in a deployment closer to the BS. This proximity improves the channel conditions of the straggler, leading to a higher power for local gradient transmission of the other WDs, thereby enhancing the received echo SNR for local target distance estimation. 
In contrast, under higher sensing gains, positioning coverage increases with the number of WDs. This is because an increased number of WDs provides more valid local distance estimates, which improves the accuracy of the clustering-based positioning algorithm at the BS. In addition, the larger number of devices enhances geometric diversity, thereby reducing the HDOP.

Furthermore, we evaluate the positioning coverage under different sensing gains in Fig. \ref{positioning_scheme}, where the number of WDs is set to 20 and the minimum inter-WD distance $d_{mid}$ is 15 m. We compare the proposed zero-overhead sensing method with four benchmark schemes. In addition to power-splitting-based and user-splitting-based ISCC methods, we further consider the proposed method without the local majority selection process, where each WD directly uploads the estimated distance $\ddot{R}_{m,t}$ at the current FL round by skipping step (\ref{majority_select}). Besides, we consider the proposed method without density-based clustering, where the BS directly applies least-squares estimation in step (\ref{q_set}) using the reported distances and known coordinates of all WDs. 
It is observed that the proposed method outperforms all benchmark schemes, especially under the high sensing-gain regimes, achieving 38.0\% and 41.5\% higher positioning coverages compared to the schemes without local majority selection and without density-based clustering, respectively, when the sensing gain is 30 dB. This highlights the importance of selecting robust local distance estimates and detecting outliers for centralized positioning decisions in improving sensing performance under various system conditions. 
\textcolor{black}{Meanwhile, power-splitting-based ISCC and user-splitting-based ISCC only achieve 5.5\% and 0.1\% positioning coverages, respectively,  when the sensing gain is 30 dB. This is because the resource allocated to the sensing module in traditional ISCC approaches is limited and cross-module interference severely degrades the sensing accuracy. In contrast,  the proposed sensing-native design substantially improves the positioning coverage through  FL gradient signal reuse without resource competition and interference-aware sensing design.}

%Power-splitting-based ISCC relies on dedicated sensing waveforms that compete with FL signals for transmit power, while it also lacks the ability to exclude unreliable WDs from the localization process. The User-splitting-based ISCC performs even worse because the dual-function WDs still suffer from inter-device interference without interference-aware sensing design. As a result, increasing the sensing gain alone cannot effectively eliminate the unreliable sensing observations. These results demonstrate that the proposed sensing-native design substantially improves sensing performance by reusing native FL gradient signals, avoiding resource competition, and filtering unreliable sensing observations.

\section{Conclusions}\label{sec_con}
This paper proposed a novel sensing-native over-the-air FL framework by seamlessly integrating wireless sensing functionality into per-round over-the-air model aggregation without compromising FL performance. Leveraging the exceptional auto-correlation properties of over-the-air FL waveform and ready-made learning statistics signals to deliver local sensing results, we embed sensing capabilities without requiring additional time or frequency resources. We developed a zero-overhead cooperative localization method by incorporating matched-filtering-based local distance estimation and robust trilateration with density-based clustering at the BS. By explicitly capturing the convergence performance under imperfect model aggregation and noisy gradient-statistics transmission, we proposed a statistics-aware communication-learning co-design method to improve the learning performance of the sensing-native over-the-air FL. 
%We obtained the optimal transmit power budgets allocated to local gradients and their statistics in closed forms, based on which an efficient SCA-based method was proposed to optimize the receiver beamforming. 
Simulation results demonstrated that the proposed framework achieves superior learning performance while simultaneously improving sensing coverage and localization accuracy. 
%This work revealed the potential of over-the-air FL as a promising sensing-native paradigm for ISCC.

\appendices
\section{Proof of Lemma \ref{Le_error_bound}}\label{proof_L1}
Let $e_{t}[d]$ denote the $d$-th element of per-slot aggregation error $\mathbf{e}_t$. Then, we have 
\begin{align}
    \mathbb{E}&\left[|e_{t}[d]|^2\right] \notag \\
    &= \frac{1}{\left(\sum_{m=1}^{M}K_m\right)^2}\mathbb{E}\Bigg[\Bigg|\sum_{m=1}^{M}K_m (g_{m,t}[d] -\hat{\mu}_{m,t})  \notag \\
    &\quad - \frac{1}{\sqrt{\eta_t}}\!\left(\!\sum_{m=1}^{M} \!\mathbf{f}^H \mathbf{h}_{bs,m} x_{1,m,t}[d]\!+\!\mathbf{f}^H\boldsymbol{\omega}_{1,t}\!\right)\!\!\Bigg|^2 \!\Bigg] \notag \\
    &= \frac{1}{\left(\sum_{m=1}^{M}K_m\right)^2}\mathbb{E}\Bigg[\!\bigg|\!\sqrt{2}\!\sum_{m=1}^{M}\!\mu_{\text{max}}\text{Re}\!\left\{\!\frac{\mathbf{f}^H\boldsymbol{\omega}_{1,t}}{\mathbf{f}^H\mathbf{h}_{bs,m}\sqrt{p_{3,m,t}}}\!\right\}\notag \\
    &\quad\!+\!\!\sum_{m=1}^{M}\!\!\left(\!\!K_m\!\!-\!\!\frac{\mathbf{f}^H\mathbf{h}_{bs,m}p_{1,m,t}}{\sqrt{\eta_t}\nu_{m,t}}\!\right)\!(g_m[d]\!-\!\mu_{m,t})\!+\!\frac{\mathbf{f}^H\!\boldsymbol{\omega}_{1,t}}{\sqrt{\eta_t}}\!\bigg|^2 \!\Bigg] \notag \\
    &\overset{(a)}{=} \frac{1}{\left(\sum_{m=1}^{M}\!\!K_m\!\right)^2\!}\Bigg( \frac{\sigma^2}{\mathbb{E}[\eta_t]}+\sum_{m=1}^M \frac{\mu_{\text{max}}^2K_m^2\sigma^2}{p_{3,m,t}|\mathbf{f}^H\mathbf{h}_{bs,m}|^2} \notag \\
    &\quad\!+\!\mathbb{E}\!\Bigg[\!\bigg|\!\sum_{m=1}^M \!\!\left(\!\!K_m\!\!-\!\!\frac{\mathbf{f}^{\!H}\mathbf{h}_{bs,m}p_{1,m,t}}{\sqrt{\eta_t}\nu_{m,t}}\!\!\right)\!\!(g_m[d]\!-\!\mu_{m,t})\bigg|^2 \Bigg]\!\Bigg),
\end{align}
where the equality (a) is due to $||\mathbf{f}||_2^2=1$ and the independence of communication noise. Accordingly, we have 
\begin{small}
\begin{align}\label{err2_a}
    \mathbb{E}&[||\mathbf{e}_{t}||_2^2] = \sum_{d=1}^{D}\mathbb{E}\left[\left| e_{t}[d]\right|^2\right] \notag \\
    &= \frac{1}{\left(\sum_{m=1}^{M}K_m\right)^2}\Bigg( \frac{D\sigma^2}{\mathbb{E}[\eta_t]}+ D\sigma^2\sum_{m=1}^M \frac{\mu_{\text{max}}^2K_m^2}{p_{3,m,t}|\mathbf{f}^H\mathbf{h}_{bs,m}|^2} \notag \\
    &\quad\!\!\!+\!\!\underbrace{\mathbb{E}\!\Bigg[\!\bigg|\!\!\sum_{m=1}^M \!\!\!\left(\!\!K_m\!\!-\!\!\frac{\mathbf{f}^H\mathbf{h}_{bs,m}p_{1,m,t}}{\sqrt{\eta_t}\nu_{m,t}}\!\!\right)\!\!\sum_{d=1}^D(g_m[d]\!\!-\!\!\mu_{m,t})\bigg|^2 \!\Bigg]}_{T_1}\!\!\Bigg)\!,
\end{align}
\end{small}
Ideally, one can minimize $\mathbb{E}[||\mathbf{e}_{t}||_2^2]$ by setting $p_{1,m,t} = \frac{K_m\sqrt{\eta_t}{\nu}_{m,t}}{\mathbf{f}^H\mathbf{h}_{bs,m}}$ for $T_1=0$ in (\ref{err2_a}). Nevertheless, in the proposed sensing-native over-the-air FL system, with only estimated statistics $\hat{\nu}_{m,t}$ available at the BS, the transmit power for each WD's local gradient is given by (\ref{pg}).
Given the transmit power constraint $\mathbb{E}[|x_{1,m,t}|^2] = |p_{1,m,t}|^2 \leq P_g$ for the local gradient, the normalization factor at the BS should be set to (\ref{eta}) in order to minimize $\mathbb{E}[||\mathbf{e}_{t}||_2^2]$ in (\ref{err2_a}).

By substituting (\ref{r2_trans}), (\ref{r3_trans}), (\ref{pg}), and (\ref{eta}) into (\ref{err2_a}), we have 
\begin{small}
\begin{align}\label{err2_c}
    \mathbb{E}&[||\mathbf{e}_{t}||_2^2] \notag \\
    &{\leq} \frac{1}{\left(\sum_{m=1}^{M}K_m\right)^2}\Bigg( \frac{D\sigma^2}{\mathbb{E}[\eta_t]}+ D\sigma^2\sum_{m=1}^M \frac{\mu_{\text{max}}^2K_m^2}{p_{3,m,t}|\mathbf{f}^H\mathbf{h}_{bs,m}|^2} \notag \\
    &\quad+\!\sigma^2\!\sum_{m=1}^M \!\frac{\nu_{\text{max}}^2K_m^2}{p_{2,m,t}|\mathbf{f}^H\mathbf{h}_{bs,m}|^2}\times\frac{\sum_{d=1}^{D}(g_m[d]\!-\!\mu_{m,t})^2}{\nu_{m,t}^2}\!\Bigg) \notag \\
    &= \frac{D\sigma^2} { K^2} \!\left[\left(\! \frac{\mu_{\text{max}}^2}{P_\mu}\!+\!\frac{\nu_{\text{max}}^2}{P_\nu}\! \right)\!\sum_{m=1}^M \! \frac{K_m^2}{|\mathbf{f}^H\mathbf{h}_{bs,m}|^2} 
    \!+\!\frac{1}{\mathbb{E}[\eta_t]}\right] \notag \\
    &{\leq}\frac{D\sigma^2} { K^2}\left( \frac{\mu_{\text{max}}^2}{P_\mu}+\frac{\nu_{\text{max}}^2}{P_\nu} \right) \sum_{m=1}^M  \frac{K_m^2}{|\mathbf{f}^H\mathbf{h}_{bs,m}|^2}  \notag \\
    &\quad+\frac{D\sigma^2}{P_g K^2}\max_{m\in M}\left(\frac{K_m^2}{|\mathbf{f}^H\mathbf{h}_{bs,m}|^2}\left(\nu_{m,t}^2\!+\!\frac{\nu_{\text{max}}\sigma^2}{P_\nu|\mathbf{f}^H\mathbf{h}_{bs,m}|^2}\right)\right) \notag \\
    &{\leq}\frac{D\sigma^2} { K^2} \left( \frac{\mu_{\text{max}}^2}{P_\mu}+\frac{\nu_{\text{max}}^2}{P_\nu} \right)\!\sum_{m=1}^M  \frac{K_m^2}{|\mathbf{f}^H\mathbf{h}_{bs,m}|^2}  \notag \\
    &\quad+\frac{D^2\sigma^2\nu_{\text{max}}^2} {P_G K^2} \max_{m\in M}\frac{K_m^2}{|\mathbf{f}^H\mathbf{h}_{bs,m}|^2} \notag \\
    &\quad+\frac{D^2\sigma^4\nu_{\text{max}}^2} {P_GP_\nu K^2} \max_{m\in M} \frac{K_m^2}{|\mathbf{f}^H\mathbf{h}_{bs,m}|^4}.
\end{align}
\end{small}
This completes the proof. 

\section{Proof of Theorem \ref{Th_gap}}\label{proof_T1}
With the assumption of Lipschitz smoothness and $\mathbf{w}_{t+1} = \mathbf{w}_t - \frac{1}{L} (\nabla F(\mathbf{w}_t)-\mathbf{e}_t)$, we have
\begin{small}
\begin{align}\label{proof_1}
    \mathbb{E}&[F(\mathbf{w}_{t+1})] - \mathbb{E}[F(\mathbf{w}_t)] \notag \\
    &\leq -\frac1L(\nabla F(\mathbf{w}_t)-\mathbf{e}_t)^T\nabla F(\mathbf{w}_t)+\frac1{2L}||\nabla F(\mathbf{w}_t)-\mathbf{e}_t||_2^2 \notag \\
    &\leq -~\frac1L~||\nabla F(\mathbf{w}_t)||_2^2+\frac{1}{L}\mathbf{e}_t^T\nabla F(\mathbf{w}_t) \notag \\
    &\quad+\frac{1}{2L}||\nabla F(\mathbf{w}_t)||_2^2-\frac{1}{L}\mathbf{e}_t^T\nabla F(\mathbf{w}_t)+\frac{1}{2L}||\mathbf{e}_t||_2^2 \notag \\
    &\leq -\frac1{2L}||\nabla F(\mathbf{w}_t)||_2^2+\!\frac 1 {2L} \left( \mathbb{E}\left[||\mathbf{e}_{t}||_2^2\right] \right),
\end{align}
\end{small}
where the expectation operator is taken with respect to the randomness of the communication noise.

With \textit{Assumption \ref{PL-inequality}}, i.e., $||\nabla F(\mathbf{w}_t)||_2^2 \geq 2S(F(\mathbf{w}_{t})-F(\mathbf{w}^*))$, we have 
\begin{align}\label{single_step_2}
    \mathbb{E}&[F(\mathbf{w}_{t+1})] - \mathbb{E}[F(\mathbf{w}_t)] \notag \\
    &\leq \frac 1 {2L} \left( \mathbb{E}\left[||\mathbf{e}_{t}||_2^2\right] \right)-\!\frac{S}{L}\left(F(\mathbf{w}_{t})-F(\mathbf{w}^*)\right).
\end{align}
By subtracting $F(\mathbf{w}^*)$ from both sides and applying (\ref{single_step_2}) recursively with $t$ iterations, we have
\begin{align}\label{single_step_end}
    \mathbb{E}&\left[F(\mathbf{w}_{t+1})-F(\mathbf{w}^*)\right] \notag \\
    &\leq\!\left(1 - \frac{S}{L}\right)\mathbb{E}\left[F(\mathbf{w}_t)-F(\mathbf{w}^*)\right] \!+\! \frac 1 {2L} \left( \mathbb{E}\left[||\mathbf{e}_{t}||_2^2\right] \right)\notag \\
    &\leq\!\Psi^{t+1}\!\left(F(\mathbf{w}_0)-F(\mathbf{w}^*)\right) \!+\! \frac{1}{2L}\! \sum_{k=0}^{t} \Psi^{t-k} \mathbb{E}[||\mathbf{e}_{k}||_2^2],
\end{align}
where $\Psi = 1 - \frac{S}{L}$.

By applying \textit{\textbf{Lemma \ref{Le_error_bound}}} and $\sum_{t} \Psi^t = \frac{1-\Psi^{t+1}}{1-\Psi}$, we have (\ref{theorm1}).

\section{Proof of Propositions \ref{prop:P_mu_convex}-\ref{prop:P_g_p_nu}}
\subsection{Optimal Power Budget Allocation for Local Gradient and Its Variance}\label{app:proof_pnu_pg}
Notice that the objective function of Problem (\textbf{P1}) is monotonically decreasing with respect to $P_\nu$. Then, given $P_G$ and $P_\mu$, the optimal power budget allocated to the local gradient's variance is 
\begin{align}\label{eq:app_Pnu}
    P_{\nu}^* = P_{\text{max}} - P_G - P_\mu.
\end{align}
With optimal $ P^*_\nu$, fixed $P_\mu$ and $\mathbf{f}$, the optimization problem over transmit power budget $P_G$ is given by
\begin{align}\label{eq:app_P2}
    \min_{P_G}~ G(P_G) \triangleq \frac{b}{\Delta - P_G}+\frac{c}{P_G\left(\Delta - P_G\right)}+\frac{d}{P_G},
\end{align}
where $\Delta=P_{\text{max}}-P_\mu$ and $b, c, d$ are defined in (\ref{eq:derivative_expression}), respectively. By calculating the partial derivative and letting $\frac{\partial G}{\partial P_G}=0$, we obtain the unique optimal $P_G^*$ in closed form:
\begin{align}\label{eq:app_PG_star}
    P_{G}^* =\Delta \left(\frac{\sqrt{c+d\Delta}}{\sqrt{c+d\Delta}+\sqrt{c+b\Delta}}\right) = \epsilon \Delta,
\end{align}
where $\epsilon$ is defined in (\ref{eq:scale}). By substituting (\ref{eq:app_PG_star}) into (\ref{eq:app_Pnu}), we obtain optimal $P_\nu^*$ in (\ref{eq:opt_Pnu_closed}).

\subsection{Optimal Power Budget Allocation for Local Gradient's Mean}\label{app:proof_pmu}
With the optimal $P_G^*$ and $P^*_\nu$ in (\ref{eq:opt_PG_closed}) and (\ref{eq:opt_Pnu_closed}), the optimization over $P_\mu$ for Problem (\textbf{P1}) is given by 
\begin{align}\label{eq:app_P3}
    \min_{P_\mu}~J(P_\mu) = \frac{a}{P_\mu} + \frac{1}{Q^2(\Delta)},
\end{align}
where $Q(\Delta) = \frac{\sqrt{c+b\Delta}-\sqrt{c+d\Delta}}{b-d}$. To find the optimal $P_\mu$, we compute the derivative with respect to $P_\mu$, i.e.,
\begin{align}
    \frac{\partial J(P_\mu)}{\partial P_\mu} =\hat{J}(P_\mu) =  -\frac{a}{P_\mu^2} + \frac{2}{Q^3(\Delta)} \frac{\partial Q(\Delta)}{\partial \Delta}.
\end{align}
We then calculate the second derivative 
\begin{align}
    \frac{\partial^2 \!J(P_\mu)}{\partial P_\mu^2}\!=\!\frac{2a}{P_\mu^3} \!+\! \frac{6}{Q^4(\Delta)}\! \left(\! \frac{\partial Q(\Delta)}{\partial \Delta}\!\! \right)^2 \!\!\!\!-\! \frac{2}{Q^3(\Delta)} \frac{\partial^2 Q(\Delta)}{\partial \Delta^2},
\end{align}
where $\frac{\partial^2 \!Q(\Delta)}{\partial \Delta^2} \!\!=\!\! \frac{1}{4(b\!-\!d)} \!\!\left[\! \!\frac{1}{\left(\frac{c}{d^{4/3}} \!+\! \frac{\Delta}{d^{1/3}}\right)^{3/2}} \!-\! \frac{1}{\left(\frac{c}{b^{4/3}} \!+\! \frac{\Delta}{b^{1/3}}\right)^{3/2}}\! \!\right]\!\!.$
%\begin{align}   
%\end{align}
For both cases where $b > d$ and $b < d$, we have $\frac{\partial^2 Q(\Delta)}{\partial \Delta^2} < 0$. Accordingly, given that $\frac{2a}{P_\mu^3} > 0$ and $Q(\Delta) > 0$ for $\Delta > 0$, we have $\frac{\partial^2 J(P_\mu)}{\partial P_\mu^2} > 0$, which shows that $\frac{\partial J(P_\mu)}{\partial P_\mu}$ is monotonically increasing with respect to $P_\mu$.

When $P_\mu \to 0^+$, we have $-\frac{a}{P_\mu^2} \to -\infty$, leading to $\frac{\partial J(P_\mu)}{\partial P_\mu} \to -\infty$. When $P_\mu \to P_{\text{max}}^-$, we have $\Delta \to 0^+$, which results in $Q(\Delta) \to 0^+$ and $\frac{\partial J(P_\mu)}{\partial P_\mu} \to +\infty$. Therefore, there exists a unique root $P_\mu^* \in (0, P_{\text{max}})$ satisfying $\frac{\partial J(P_\mu)}{\partial P_\mu}=0$. 

\section{Proof of Proposition \ref{beamforming}}\label{proof_dual}
Given the Lagrangian in (\ref{lagrangian}), at the optimum, the stationarity condition requires that $\frac{\partial \mathcal{L}}{\partial \mathbf{f}}=\mathbf{0}$, i.e.,
\begin{align}
\sum_{m=1}^{M} (a'+b')K_m^2 \left(-\frac{2\mathbf{H}_{bs,m}\mathbf{f}^*}{\Phi_m^2(\mathbf{f}^*)}\right) &- \sum_{m=1}^{M} C_m^* (2\mathbf{H}_{bs,m}\mathbf{f}^*) \notag \\
- \sum_{m=1}^{M} B_m^* (4\Phi_m(\mathbf{f}^*) \mathbf{H}_{bs,m}\mathbf{f}^*) &+ 2A^*\mathbf{f}^* = \mathbf{0}.
\end{align}
Equivalently, we have 
\begin{align}
\left( \!\sum_{m=1}^{M} \!\!\left[ \!\frac{(a'\!\!+\!b')K_m^2}{\Phi_m^2(\mathbf{f}^*)} \!\!+ \!\!2B_m^* \Phi_m(\mathbf{f}^*) \!\!+\!\! C_m^* \!\right]\! \mathbf{H}_{bs,m}\!\! \right)\! \mathbf{f}^* \!\!=\!\! A^* \mathbf{f}^*\!.
\end{align}
\bibliographystyle{IEEEtran}
\bibliography{IEEEabrv,myref}

\end{document}